\documentclass[a4paper,12pt]{article}
\usepackage{a4wide}
\usepackage{amssymb, amsmath}
\usepackage[amsmath,hyperref,thmmarks]{ntheorem}
\usepackage{t1enc,times,graphicx}
\usepackage[utf8]{inputenc}
\usepackage[english]{babel}
\usepackage{enumerate}
\usepackage{hyperref}
\usepackage{color,caption,subcaption}
\usepackage{natbib}
\graphicspath{{results2/pnl}}

\emergencystretch=\hsize
\def\ca{{\mathcal A}}

\def\cf{{\mathcal F}}

\def\ch{{\mathcal H}}

\def\ct{{\mathcal T}}

\def\E{{\mathbb E}}
\def\H{{\mathbb H}}

\def\L{{\mathbb L}}
\def\N{{\mathbb N}}
\def\P{{\mathbb P}}

\def\R{{\mathbb R}}

\def\s{\star}

\def\ind#1{{\bf 1}_{\left\{#1\right\}}}

\def\Var{\mathop{\rm Var}\nolimits}

\def\inv#1{\mathop{\frac{1}{ #1}}\nolimits}

\theoremstyle{plain}
\newtheorem{theorem}{Theorem}[section]
\newtheorem{proposition}[theorem]{Proposition}

\newtheorem{corollary}[theorem]{Corollary}
\newtheorem{remarkh}[theorem]{Remark}
\newenvironment{remark}{\begin{remarkh}\rm}{\end{remarkh}}

\theoremstyle{nonumberplain}
\theoremheaderfont{\normalfont\itshape}
\theorembodyfont{\normalfont}
\theoremseparator{.}
\theoremsymbol{\ensuremath{\blacksquare}}
\newtheorem{proof}{Proof}

\makeatletter
\newcounter{hypo}
\renewcommand{\thehypo}{(${\mathcal H}$-\arabic{hypo})}
\newcommand{\dohypo}{\thehypo}

\makeatother
\title{
A pure dual approach for hedging Bermudan options}
\date{\today}
\author{Aurélien Alfonsi\thanks{CERMICS, Ecole des Ponts, Marne-la-Vall\'ee, France. MathRisk, Inria, Paris,
France. \newline \texttt{aurelien.alfonsi@enpc.fr}} \and Ahmed Kebaier\thanks{LaMME, CNRS, UMR 8071, Université d'Evry, Université Paris-Saclay, 91037, Evry, France. \newline \texttt{ahmed.kebaier@univ-evry.fr}} 
\and Jérôme Lelong\thanks{ Univ. Grenoble Alpes, CNRS, Grenoble INP, LJK, 38000 Grenoble, France. \newline\texttt{jerome.lelong@univ-grenoble-alpes.fr} \newline{Acknowledgements: AA and AK benefited from the support of the “chaire Risques financiers”, Fondation du Risque.}}}

\begin{document}
\maketitle
\begin{abstract}
  This paper develops a new dual approach
  to compute the hedging portfolio of a Bermudan option and its initial value. It gives a ``purely dual'' algorithm following the spirit of~\cite{rogers-10} in the sense that it only relies on the dual pricing formula. The key is to rewrite the dual formula as an excess reward representation and to combine it with a strict convexification technique. The hedging strategy is then obtained by using a Monte Carlo method, solving backward a sequence of least square problems. We show convergence results for our algorithm and test it on many different Bermudan options. Beyond giving directly the hedging portfolio, the strength of the algorithm is to assess both the relevance of including financial instruments in the hedging portfolio and the effect of the rebalancing frequency.
\end{abstract}

%{\bf NOTATIONS: $q$ -> sample Monte-Carlo, $k$ -> asset, $\ell$ -> base de projection, $n,i,j$ -> temps}\\

\section*{Introduction}

This paper proposes a new algorithm to compute the hedging portfolio of a Bermudan option and its corresponding price. American and Bermudan options have been widely studied in the literature. The algorithm proposed by~\cite{LS} is probably the most widely used in the financial industry for pricing these options by Monte Carlo. It consists first in approximating the optimal stopping rule and then in using it to evaluate the Bermudan claim. It is a pure primal algorithm that gives the value from the option buyer point of view. {Many different algorithms have been proposed for the valuation of American or Bermudan claims by Monte-Carlo. Here, we mention in particular~\cite{HK,AB,BMSp,DFM} who aimed at giving both upper and lower bounds for the Bermudan price.}  However, in the theory of Black, Scholes and Merton, the price of an option corresponds to the initial value of a self-financing portfolio which exactly hedges the option. In the case of Bermudan options, it is not clear for an option seller how to construct a practical hedging portfolio from the single Longstaff and Schwartz price { or, more generally, from the prices obtained by the different methods mentioned above. }

{Several papers such as~\cite{WC,BMSc} have proposed Monte-Carlo methods to compute the delta and other sensitivities of Bermudan options that can be used then for hedging. In a different direction, }
 \cite{rogers-02} and~\cite{HK} have introduced a dual formula that directly puts in evidence the hedging portfolio by the mean of a martingale. \cite{rogers-02} derived from this formula an algorithm which provides a hedging portfolio, but the algorithm is however limited to some specific payoff functions (see the discussion in Subsection~\ref{subsec_num1d} for an illustration). {\cite{DFM} have proposed a more generic Monte-Carlo approach with a finite vector subspace of approximating martingales, so that the dual formula boils down to a Linear Program (LP). Even though LP can be solved very efficiently, there is a high number of constraints that restricts in practice the number of samples and/or the dimension of the martingale approximating space.  However in the literature, the dual formula has mostly been used to produce an upper bound of the Bermudan option price} as initiated by~\cite{HK} and later on~\cite{AB}, while the Longstaff Schwartz algorithm provides a lower bound.  To be more precise, \cite{AB} used the Longstaff Schwartz algorithm to construct a suitable martingale from the approximated optimal stopping time. This martingale produces their upper bound price.  However, it is again not obvious how to construct a practical hedging portfolio that corresponds to this martingale by using tradable financial instruments. This issue was pointed out by~\cite{rogers-10} who advocated a {\em purely dual} approach, i.e. an algorithm that is only based on the dual formula. This is precisely the goal of our paper: we develop a purely dual algorithm that directly calculates the hedging portfolio of a Bermudan option that can be built using only tradable financial instruments.

To obtain our algorithm, we derive two key results from the dual formula~\cite[Theorem 2.1]{rogers-02}. First, we rewrite the expectation as a telescopic sum, which enables us to decompose the minimisation problem into a sequence of minimisation problems on the martingale increments.  {This decomposition is a key difference with the LP algorithm of~\cite{DFM} which directly solves the global minimisation problem. } Second, we introduce a strict convexification of these problems, which gives the uniqueness of the solution and makes the resolution of the problem easily tractable. These two steps are explained in Section~\ref{Sec_main}. The remainder of the paper is organized as follows. Section~\ref{Sec_NumericalAnalysis} analyses the approximation error of our purely dual algorithm and shows the convergence of the related Monte Carlo estimators. Section~\ref{Sec_Implementation} presents the implementation of our algorithm in a financial framework and make precise the martingales that are used in the hedging strategy. Last, Section~\ref{Sec_Num} gives a detailed presentation of our numerical results for the hedging and valuation of different Bermudan options on one or several assets. We show the relevance of the algorithm, {comparing it with some existing delta hedging methods}, and reporting the P\&L of discrete hedging portfolios calculated with it. Namely, we present hedging portfolios with different rebalancing frequencies that are either constructed with the underlying assets only or with adding vanilla European options. Of course, a hedging portfolio with a higher rebalancing frequency and with more financial instruments can in principle do better. However, this requires higher computational efforts and also in practice higher operational costs. Thus, it is interesting to assess the relevance of increasing the rebalancing frequency or of adding financial instruments in the hedging portfolio. We observe on our examples that there is no absolute answer. Adding European options in the hedging portfolio improves the hedge in most of our cases but one, while increasing the rebalancing frequency is relevant up to a certain value, depending on the option considered. Thus, beyond giving an upper bound price of the Bermudan option and its associated hedging portfolio, our algorithm is useful
\begin{itemize}
  \item to value the interest of increasing the rebalancing frequency of the hedging portfolio.
  \item to select the more relevant instruments to be included in the hedging portfolio.
\end{itemize} 
Up to our knowledge, these two features are new in the literature and relevant for practitioners.

\section{Framework and main result}\label{Sec_main}

We introduce the notation of a classical optimal stopping problem in discrete time. Let $N\in \N^*$ and $(\Omega, \overline{\cf}, \cf=(\cf_n)_{0 \le n \le N}, \P)$ be a filtered probability space with a discrete time filtration. Let $p\ge 1$. We consider a family $(Z_n)_{0\le n\le N}$ of $\cf$-adapted real valued random variables such that $\max_{0\le n\le N}\E[|Z_n|^p]<\infty$. 
We define, for $n\in \{0,\dots,N\}$,
\begin{equation}
  \label{eq:price}
  U_{n} = \sup_{\tau \in \ct_{n, N}} \E[ Z_{\tau} | \cf_{n}],
\end{equation}
where $\ct_{n,N}$ denotes the set of $\cf$-stopping times taking values in $\{n,\dots,N\}$.
Using the Snell envelope theory, the sequence $U$ can be proved to solve
the following dynamic programming equation
\begin{equation}
  \label{eq:dpp}
  \begin{cases}
    U_{N} & = Z_{N} \\
    U_{n} & = \max\left( Z_{n}, \E[U_{n+1} | \cf_{n}] \right), \quad 0 \le n  \le N-1,
  \end{cases}
\end{equation}
see~\cite[Section 2.2]{LL} for instance. The process $U$ is a $L^p$ supermartingale,  $\tau^\s:=\inf \{n \in \{0,\dots, N\}: U_n=Z_n \}$ is an optimal stopping time for~\eqref{eq:price}, and $(U_{n\wedge \tau^\s})_{0 \le n \le N}$ is a martingale.

Alternatively, $U$ can also admit a dual representation as proposed by~\cite{rogers-02,rogers-10} and~\cite{HK}. Denoting $\H^p$ the set of $\cf$-martingales that are $L^p$ integrable, we have 
\begin{align}
  \label{eq:dual-price-k}
  U_{n} & = \inf_{M \in \H^p} \E\left[ \max_{n \le j \le N} \{Z_{j} - (M_{j}- M_{n})\}\bigg| \cf_{n}\right].
\end{align}
Indeed, we have by the optional stopping theorem
\begin{equation}\label{eq:ineg_U} U_{n} = \sup_{\tau \in \ct_{n, N}} \E[ Z_{\tau} | \cf_{n}]= \sup_{\tau \in \ct_{n, N}} \E[ Z_{\tau}- (M_{\tau}- M_{n}) | \cf_{n}]\le \E\left[ \max_{n \le j \le N} \{Z_{j} - (M_{j}- M_{n})\}\bigg| \cf_{n}\right], \end{equation}
for any $M\in \H^p$. For the other inequality, we introduce the Doob-Meyer decomposition
\begin{equation}
  \label{eq:doob-meyer}
  U_{n} = U_0 + M^\s_{n} - A^\s_{n},
\end{equation}
where $M^\s\in \H^p$ vanishes at $0$ and $A^\s$ is a predictable, nondecreasing and $L^p$-integrable process. More precisely, 
we have for $ 1 \le n \le N$
\begin{equation}\label{decompo_DM}
  M^\s_{n}-M^\s_{n-1}=U_{n}-\E[U_{n}|\cf_{n-1}], \  A^\s_{n}-A^\s_{n-1}=U_{n-1}-\E[U_{n}|\cf_{n-1}].
\end{equation}
The Doob-Meyer martingale $M^\s$ achieves the infimum in~\eqref{eq:dual-price-k}. This can be checked by noting that $Z_n\le U_n$ and 
\begin{align*}
  \max_{n \le j \le N} \{Z_{j} - (M^\s_{j}- M^\s_{n}) \}&\le \max_{n \le j \le N} \{U_{j} - (M^\s_{j}- M^\s_{n}) \} \\&= \max_{n \le j \le N} \{U_0-A^\s_j+ M^\s_{n} \}= U_0-A^\s_n+ M^\s_{n}=U_n,
\end{align*}
since $A^\s$ is nondecreasing. 
From the last inequality and~\eqref{eq:ineg_U}, we get~\eqref{eq:dual-price-k}. Besides, we get in addition that
\begin{equation}\label{eq:almost_sure_Mstar} U_n= \max_{n \le j \le N} \{Z_{j} - (M^\s_{j}- M^\s_{n}) \}.
\end{equation}
It is said that $M^\s$ is surely optimal for every $n \in \{0,\dots,N\}$. This property has been studied in detail by~\cite{schoen12-1}.   

\begin{remark}\label{Rk_hedging}
  Let $M \in \H^p$ be a martingale such that $M_0=0$ and let $V_0=\E[\max_{0\le n \le N} \{Z_n-M_n\}]\ge U_0$. In mathematical finance, $V_0+M_n$ can be interpreted as the value at time~$n$ of a self-financing portfolio. Thus, we may be interested in quantifying the pricing error $V_0-U_0$ and  the $L^p$ hedging error $\E[|Z_{\tau^\s}-(V_0+M_{\tau^\s})|^p]^{1/p}$ when using this portfolio. The first one is upper bounded by
  $$0\le V_0-U_0\le  \E[  |\max_{0\le n\le N} M^\s_{n}-M_{n} |], $$
  and then we can use Burkholder-Davis-Gundy inequality. For the second one, we observe that
   $Z_{\tau^\s}=U_{\tau^\s}=U_0+M^\s_{\tau^\s}$ and get
  \begin{align*}
    \E[|Z_{\tau^\s}-(V_0+M_{\tau^\s})|^p]^{1/p}&\le |U_0-V_0|+ \E[|M^\s_{\tau^\s}-M_{\tau^\s}|^p]^{1/p} \\
    & \le V_0- U_0 +\E[|M^\s_{N}-M_{N}|^p]^{1/p}.\end{align*}
  For $p=2$, using Cauchy-Schwarz and Doob's inequalities, we get
  $V_0-U_0\le 2 \E[|M^\s_{N}-M_{N}|^2]^{1/2}$ and 
  \begin{equation}\label{L2_hedging_error}
    \E[|Z_{\tau^\s}-(V_0+M_{\tau^\s})|^2]^{1/2}\le  3 \E[|M^\s_{N}-M_{N}|^2]^{1/2}.\end{equation}

\end{remark}

\subsection{Excess reward representation of the optimal stopping problem}

 We start by rewriting the dual formula~\eqref{eq:dual-price-k} of the optimal stopping problem. 
For $n \in \{1,\dots,N\}$, we define the spaces $\ch^p_n$ of $L^p$~martingale increments
\[
  \ch^p_n = \{ Y \in \L^p(\Omega) \,:\, Y \text{ is real valued, }\cf_{n}-\text{mesurable and } \E[Y | \cf_{n-1}] = 0\}.
\]
Let $M \in \H^p$. We can write  for every $1 \le n \le N$, $\Delta M_n=M_{n}-M_{n-1} \in \ch^p_n$. With this new notation, we obtain the telescopic decomposition
\begin{align*}
  &  \max_{0 \le j \le N} \{Z_{j} - (M_{j} -M_0)\}  \\
  & = Z_N-(M_N-M_0) + \sum_{n=0}^{N-1}\left( \max_{n \le j \le N} \{Z_{j} - M_{j}\}  -  \max_{n + 1\le j \le N} \{Z_{j} - M_{j}\} \right)\\
  & = Z_N-(M_N-M_0) + \sum_{n=0}^{N-1} \left( \max_{n \le j \le N} \{Z_{j} - (M_{j}-M_{n+1})\}  -  \max_{n + 1\le j \le N} \{Z_{j} - (M_{j}-M_{n+1})\} \right)\\
  & = Z_N-(M_N-M_0) + \sum_{n=0}^{N-1}  \left(Z_{n} + \Delta M_{n+1} - \max_{n + 1\le j \le N} \left\{Z_{j} - \sum_{i=n+2}^{j} \Delta M_i\right\} \right)_+,
\end{align*}
using that $\max(a,b)-b=(a-b)_+$ for the last equality.
In the context of an optimal stopping problem, we get from~\eqref{eq:almost_sure_Mstar} that $$\E\left[\max_{n \le j \le N} \{Z_{j} - M^\s_{j}\}  -  \max_{n + 1\le j \le N} \{Z_{j} - M^\s_{j}\}\right]=\E[U_{n}-U_{n+1}]$$ is the expected excess reward given by the possibility of stopping at time~$n$. Taking the expectation in this formula leads to
\begin{align}
  & \E\left[ \max_{0 \le j \le N} \{Z_{j} - (M_{j}-M_0)\}\right] \notag \\
  & = \E[Z_N] + \sum_{n=0}^{N-1} \E\left[ \left(Z_{n} + \Delta M_{n+1} - \max_{n + 1\le j \le N} \left\{Z_{j} - \sum_{i=n+2}^{j} \Delta M_i\right\} \right)_+\right]. \label{eq:telescopic}
\end{align}
Therefore, we can rewrite the global minimisation as follows
\begin{align}
  U_0&= \inf_{M \in \H^p} \E\left[ \max_{0 \le j \le N} \{Z_{j} - (M_{j}-M_0)\}\right] \notag\\ 
  &=  \E[Z_N] + \inf_{\Delta M_1 \in \ch_1^p,\dots,\Delta M_N \in \ch_N^p }  \sum_{n=0}^{N-1}\E\left[ \left(Z_{n} + \Delta M_{n+1} - \max_{n + 1\le j \le N} \left\{Z_{j} - \sum_{i=n+2}^{j} \Delta M_i\right\} \right)_+\right]. \label{global_min}
\end{align}

\begin{remark}
  In the context of mathematical finance, \eqref{eq:telescopic} can be seen as the expansion of the Bermudan option price as a sum of the corresponding European option price $\E[Z_N]$ and of correcting terms $\E\left[ \left(Z_{n} + \Delta M_{n+1} - \max_{n + 1\le j \le N} \left\{Z_{j} - \sum_{i=n+2}^{j} \Delta M_i\right\} \right)_+\right]$ representing, for $M=M^\s$, the value of having the right to exercise the option at time~$n\in \{0,\dots,N-1\}$. {From~\eqref{eq:almost_sure_Mstar} and~\eqref{decompo_DM}, we get indeed for $M=M^\s$ that
  $$\left(Z_{n} + \Delta M_{n+1} - \max_{n + 1\le j \le N} \left\{Z_{j} - \sum_{i=n+2}^{j} \Delta M_i\right\} \right)_+= (Z_n-\E[U_{n+1}|\mathcal{F}_n])_+,$$ which measures the extra reward gained by exercising rather than continuing. The identity 
  $$U_0=\E[Z_N]+\sum_{n=0}^{N-1}\E[(Z_n-\E[U_{n+1}|\mathcal{F}_n])_+]$$
  can also be directly derived from the dynamic programming equation~\eqref{eq:dpp} and has been intensively used by~\cite{BMSp} to compute Bermudan price bounds. In contrast, \eqref{eq:telescopic} is derived from the dual formula and holds for any martingale~$M$. This will be essential to decompose the minimisation problem~\eqref{global_min} across the time.  }
\end{remark}

\subsection{Strict convexification and algorithm}

The value of the optimal stopping problem is given by~\eqref{eq:dual-price-k} for $n=0$. Therefore,  we are interested in computing the infimum of~\eqref{eq:telescopic} over all martingales in $\H^p$. Formula~\eqref{eq:telescopic} allows us to decompose the global minimization of~\eqref{eq:dual-price-k} into the backward minimization of \begin{equation}\label{non_cvx_pb}
  \inf_{\Delta M_{n+1} \in \ch^p_{n+1}}\E\left[ \left(Z_{n} + \Delta M_{n+1}  - \max_{n + 1\le j \le N} \left\{Z_{j} - \sum_{i=n+2}^{j} \Delta M_i\right\}\right)_+\right],  
\end{equation}
from $n=N-1$ to $n=0$. However, the non strict convexity of the positive part raises some issues in the  back propagation of the minimisation problems. To get around this issue, we will replace the positive part by a strictly convex function, as justified by the next theorem.

\begin{remark}\label{non-unicity}
  Let us emphasize that there are many solutions of the backward minimization problem~\eqref{non_cvx_pb}. Indeed, using~\cite[Theorem 4.1]{schoen12-1}, we show below in Proposition~\ref{prop_non_uniq} this result. To circumvent the non-uniqueness issue, \cite[Theorem 5.1]{schoen12-1} proposes to add a conditional variance penalization in the minimization problem, that is then no longer convex. Here, we rather use a convexification technique that makes the minimization problem strictly convex. 
\end{remark}

\begin{theorem}\label{thm:main}
  Let $\varphi:\R \to \R$ be a convex function such that $|\varphi(x)|\le C(1+|x|^p)$. Then, we have
 \begin{align}\label{pb_generalise}
  &\E[Z_N] + \sum_{n=0}^{N-1} \E\left[ \varphi\left(Z_{n} + \Delta M_{n+1} - \max_{n + 1\le j \le N} \left\{Z_{j} - \sum_{i=n+2}^{j} \Delta M_i\right\} \right)\right]\\
  &\ge \E[Z_N] + \sum_{n=0}^{N-1} \E\left[ \varphi \left(Z_{n} + \Delta M^\s_{n+1} - \max_{n + 1\le j \le N} \left\{Z_{j} - \sum_{i=n+2}^{j} \Delta M^\s_i\right\} \right)\right],\notag
 \end{align}  
 and $M^\s$ is a solution of the following sequence of backward optimization problems for $n=N-1,\dots,0$
 \begin{equation} \label{eq:inf}
 \inf_{\Delta M_{n+1} \in \ch^p_{n+1}}\E\left[ \varphi \left(Z_{n} + \Delta M_{n+1}  - \max_{n + 1\le j \le N} \left\{Z_{j} - \sum_{i=n+2}^{j} \Delta M_i\right\}\right)\right].
 \end{equation}
 When $\varphi$ is strictly convex, we have equality in~\eqref{pb_generalise} if, and only if $M=M^\s$, and $M^\s$ is the unique solution of~\eqref{eq:inf}.
\end{theorem}

\begin{proof}
We prove by backward induction that the infimum in~\eqref{eq:inf} is attained by $M^\s$. 
By Jensen's inequality, we have
\begin{align}\label{jensen_phi}
&\E\left[ \varphi \left(Z_{n} + \Delta M_{n+1}  - \max_{n + 1\le j \le N} \left\{ Z_{j} - \sum_{i=n+2}^{j} \Delta M_i \right\} \right) \Bigg| \cf_{n} \right] \\
&\ge \varphi \left( \E\left[ Z_{n} + \Delta M_{n+1}  - \max_{n + 1\le j \le N} \left\{ Z_{j} - \sum_{i=n+2}^{j} \Delta M_i \right\} \Bigg| \cf_{n} \right]\right), \notag
\end{align}
 with equality if \begin{equation}\label{cond_jensen}
Z_{n} + \Delta M_{n+1}  - \max_{n + 1\le j \le N} \left\{ Z_{j} - \sum_{i=n+2}^{j} \Delta M_i \right\} \in \cf_{n}.\end{equation} Note that the lower bound is given by $\varphi \left( \E\left[ Z_{n} - \max_{n + 1\le j \le N} \left\{ Z_{j} - \sum_{i=n+2}^{j} \Delta M_i \right\} | \cf_{n} \right]\right)$, which does not depend on $\Delta M_{n+1}$. Let us check that~\eqref{cond_jensen} is satisfied for all $n$ by $M^\s$. By~\eqref{decompo_DM} and~\eqref{eq:almost_sure_Mstar}, we have
\begin{align*}
  &   Z_{n} + \Delta M^\s_{n+1}  - \max_{n + 1\le j \le N} \left\{ Z_{j} - \sum_{i=n+2}^{j} \Delta M^\s_i \right\} \\
  &=  Z_{n} +U_{n+1}- \E[U_{n+1}|\cf_{n}]-U_{n+1}=Z_{n}- \E[U_{n+1}|\cf_{n}] \in \cf_{n},
\end{align*}
  and thus $M^\s$ satisfies~\eqref{cond_jensen} for all $0\le n \le N-1$. This shows that $M^\s$ solves~\eqref{eq:inf} and satisfies~\eqref{pb_generalise}.
  
Now, let us suppose $\varphi$ strictly convex. Then, we have equality in~\eqref{jensen_phi} if, and only if we have~\eqref{cond_jensen}.  Note that~\eqref{cond_jensen} gives $$\Delta M_{n+1} = \max_{n + 1\le j \le N}  \left\{ Z_{j} - \sum_{i=n+2}^{j} \Delta M_i \right\}- \E\left[\max_{n + 1\le j \le N} \left\{ Z_{j} - \sum_{i=n+2}^{j} \Delta M_i \right\} \Bigg|\cf_n \right],$$ which  determines $M$ uniquely by backward induction.  
Thus, $M^\s$ is the unique martingale that minimises~\eqref{pb_generalise} and solves the backward optimization problems~\eqref{eq:inf}. 
\end{proof}

\begin{remark}
  \cite[Theorem 4.6]{schoen12-1} states that if $M\in \H^p$ satisfies $\max_{n\le j \le N}  \{Z_j -(M_j-M_n) \} \in \cf_n$, then $U_n=\max_{n\le j \le N}  \{Z_j -(M_j-M_n) \} $.
Let us observe that~\eqref{cond_jensen} implies this condition since we have
\begin{align*}
  \max_{n\le j \le N}  \{Z_j -(M_j-M_n) \}&= Z_n+ \left( \max_{n+1 \le j \le N}  \{Z_j -(M_j-M_{n+1}) \} -\Delta M_{n+1}\right)_+\\
  &=Z_n+\left( \E \left[ \max_{n+1 \le j \le N}  \{Z_j -(M_j-M_{n+1}) \} \Bigg| \cf_n\right]\right)_+
\end{align*} and $Z_n\in \cf_n$.
\end{remark}

\begin{proposition}\label{prop_non_uniq}
  Let $M^\s$ and $A^\s$ be defined as in~\eqref{decompo_DM}.  Let $(\zeta_n)_{1\le n\le N}$ be any sequence of adapted random variables such that $\E[\zeta_n|\cf_{n-1}]=1$ for $n\ge 1$ and such that $\P(\zeta_n\ge 0|\cf_{n-1})=1$. Then, the martingale 
  $$ M_n = M^\s_n - A^\s_n + \sum_{j=1}^n(A^\s_n-A^\s_{n-1})\zeta_n, \ n\ge 1,$$
  solves the backward minimisation problem~\eqref{non_cvx_pb}.
\end{proposition}
\begin{proof}
 Let us first recall that an integrable random variable $\xi$ satisfies $\E[\xi_+]=\E[\xi]_+$ if  $\P(\xi\ge 0)=1$ or $\xi$ is constant. Therefore, the equality in~\eqref{jensen_phi}  with $\varphi(x)=x_+$ holds if
\begin{equation}\label{suff_cond_x+} \P\left( Z_{n} + \Delta M_{n+1}  - \max_{n + 1\le j \le N} \left\{ Z_{j} - \sum_{i=n+2}^{j} \Delta M_i \right\}\ge 0 \bigg| \cf_{n} \right) =1 \text{ or \eqref{cond_jensen}}.
\end{equation}
By~\cite[Theorem 4.1]{schoen12-1}, we have $U_n=\max_{n\le j \le N}(Z_j-M_j+M_n)$ for $0\le n\le N$, and thus 
\begin{align*} Z_{n} + \Delta M_{n+1}  - \max_{n + 1\le j \le N} \left\{ Z_{j} - \sum_{i=n+2}^{j} \Delta M_i \right\} &=Z_{n} + \Delta M_{n+1}  - U_{n+1} \\
&=Z_{n}  - \E[U_{n+1}|\cf_n]+\Delta A_{n+1}^\s (\zeta_{n+1}-1).
\end{align*}
If $\Delta A_{n+1}^\s>0$, we have $U_n=Z_n$. Since $U_n-\E[U_{n+1}|\cf_n]=\Delta A_{n+1}^\s$, we get $Z_{n}  - \E[U_{n+1}|\cf_n]+\Delta A_{n+1}^\s (\zeta_{n+1}-1)=\Delta A_{n+1}^\s \zeta_{n+1}.$
This leads to
$$Z_{n}  - \E[U_{n+1}|\cf_n]+\Delta A_{n+1}^\s (\zeta_{n+1}-1)=\mathbf{1}_{\Delta A_{n+1}^\s=0}\left(Z_n-\E[U_{n+1}|\cf_n] \right)+ \mathbf{1}_{\Delta A_{n+1}^\s>0}  \Delta A_{n+1}^\s \zeta_{n+1}.$$
Since $\Delta A_{n+1}^\s$ is $\cf_n$-measurable, we get~\eqref{suff_cond_x+}, which shows that $M$ solves~\eqref{non_cvx_pb}.
\end{proof}
\begin{remark}
 Proposition~\ref{prop_non_uniq} shows that there are, in general, many solutions to the sequence of the backward minimisation problems~\eqref{non_cvx_pb}. Then, it is not clear that a solution of this sequence  solves the global minimisation problem~\eqref{global_min}. This shows a double interest of the strict convexification: it ensures first uniqueness of the solution of each problem~\eqref{eq:inf}, which gives then precisely the unique solution of the global minimisation problem~\eqref{global_min}.   
\end{remark}

\noindent {\bf Description of the algorithm.} We now describe formally our algorithm:
\begin{enumerate}
  \item Choose a strictly convex function $\varphi:\R \to \R$  such that $|\varphi(x)|\le C(1+|x|^p)$.
  \item For each~$n\in \{1,\dots,N\}$, choose a finite dimensional linear subspace $\ch^{pr}_n$ of $\ch^p_n$.
  \item For $n=N-1$ to $n=0$, use an optimisation algorithm to  minimise 
  $$ \inf_{\Delta M_{n+1} \in \ch^{pr}_{n+1}}\E\left[ \varphi\left(Z_{n} + \Delta M_{n+1}  - \max_{n + 1\le j \le N} \left\{Z_{j} - \sum_{i=n+2}^{j} \Delta M_i\right\}\right)\right].$$
\end{enumerate} 
Unless $\ch^p_n$ is of finite dimension for all~$n$ and $\ch^{pr}_n=\ch^p_n$, this algorithm gives an approximation of the martingale~$M^\s$. Besides, in practice, the expectations can seldom be computed explicitly, and we replace them by their Monte Carlo estimators on i.i.d. samples. There are thus two sources of error: the projection error due to the approximation of $\ch^p$ by the space $\ch^p_n$ and the Monte Carlo error.   

The minimisation problems appearing in the algorithm are finite dimensional and strictly convex. Thus, they can be solved numerically by a gradient descent. However, it is more convenient to work with square integrable martingales and $\varphi(x)=x^2$, since $\Delta M_{n+1}$ solves then a classical least square problem, which avoids the use of a gradient descent. In this paper, we will work in the sequel with $p=2$ and $\varphi(x)=x^2$ and leave other choices for further investigations. From now on, we assume that $Z_n \in L^2$, and consider the backward sequence of minimisation problems:
\begin{equation}\label{eq:inf_sq} \inf_{\Delta M_{n+1} \in \ch^{pr}_{n+1}}\E\left[ \left(Z_{n} + \Delta M_{n+1}  - \max_{n + 1\le j \le N} \left\{Z_{j} - \sum_{i=n+2}^{j} \Delta M_i\right\}\right)^2\right].
\end{equation}
The next section presents a numerical analysis of the projection error and shows the convergence of the Monte Carlo estimators associated with~\eqref{eq:inf_sq}.

\begin{remark}\label{Rk_Europ}
The problem of mean-variance hedging of European claims has been widely studied in mathematical finance, see e.g.~\cite{Schweizer} for a presentation in discrete time. In our framework, it amounts to minimizing $\E\left[\left(Z_N-\sum_{n=1}^N \Delta M_n\right)^2\right]$ among (a family of) square integrable martingales~$M$. Using the martingale property, this minimisation problem boils down to minimizing $\E[\Delta M_n^2]-2\E[Z_N\Delta M_n]$ or equivalently $\E[(\Delta M_n-Z_N)^2]$, for each $n=1,\dots,N$. These minimisation problems can  be solved independently. In contrast, for the Bermudan case~\eqref{eq:inf_sq}, the minimisation problems have to be solved sequentially, backward in time. \\
Finally, let us note that this kind of minimisation problems has been recently considered by~\cite{BGTBW} to compute hedging portfolios for European options with Machine Learning. 
Thus, Problem~\eqref{eq:inf_sq} provides a simple way to extend their approach to Bermudan options. 
\end{remark}

\section{Numerical analysis of the dual algorithm}\label{Sec_NumericalAnalysis}

\subsection{Projection error}

In this subsection, we are interested in analysing the $L^2$ error between the martingale $M^\s$ and the one defined by~\eqref{eq:inf_sq} in our algorithm. We have seen in Remark~\ref{Rk_hedging} that this error gives an upper bound of the hedging error of Bermudan options. To analyse this error, we start by a remark and consider  the martingale increment $\Delta M^{\s,pr}_{n+1}$ defined as the $L^2$ projection of $\Delta M^\s$ on $\ch_{n+1}^{pr}$. Since 
$Z_{n} + \Delta M^\s_{n+1}  - \max_{n + 1\le j \le N} \left\{Z_{j} - \sum_{i=n+2}^{j} \Delta M^\s_i\right\} \in \cf_n$, we get that 
\begin{align*} %\label{eq:inf2}
  &\inf_{\Delta M_{n+1} \in \ch^{pr}_{n+1}}\E\left[ \left( Z_{n} + \Delta M_{n+1}  - \max_{n + 1\le j \le N} \left\{Z_{j} - \sum_{i=n+2}^{j} \Delta M^\s_i\right\}\right)^2\right] \\
  =&\inf_{\Delta M_{n+1} \in \ch^{pr}_{n+1}}\E\left[ \left( Z_{n} + \Delta M^\s_{n+1}  - \max_{n + 1\le j \le N} \left\{Z_{j} - \sum_{i=n+2}^{j} \Delta M^\s_i\right\}\right)^2\right] +\E[(\Delta M^\s_{n+1}-\Delta M_{n+1})^2],
  \end{align*}
so that  the unique solution of the minimization problem is $\Delta M^{\s,pr}_{n+1}$.
This means that if the martingale increments $\Delta M^\s_i$ for $i\ge n+2$ were exactly known, our minimisation would lead to the $L^2$ projection of $\Delta M^\s$ on $\ch_{n+1}^{pr}$. However, this situation only occurs in~\eqref{eq:inf_sq} for the computation of $\Delta M_N$. For $n<N$, the error made on the increments propagates through the backward recursion. This error propagation is analyzed in the next proposition.
\begin{proposition}\label{prop_errorL2}
  For $1\le n \le N$, let $\Delta M^{\s,pr}_{n}$ be the $L^2$ projection of $\Delta M^\s_n$ on $\ch_{n}^{pr}$, $\varepsilon^\s_n=\E[(\Delta M^\s_n-\Delta M^{\s,pr}_n)^2]^{1/2}$ and $\varepsilon^\s_{\max}=\max_{1\le n \le N} \varepsilon^\s_n$.
  
  There exists a unique martingale $M\in \H^2$ that solves~\eqref{eq:inf_sq} and such that $M_0=0$. Besides, if $\varepsilon_n=\E[(\Delta M^\s_n-\Delta M_n)^2]^{1/2}$ denotes the $L^2$ error between $M$ and $M^\s$ on the $n^{th}$ increment, we have $\varepsilon_N=\varepsilon^\s_N$ and 
\begin{equation}\label{rec_error} \varepsilon_{n+1}\le \varepsilon_{n+1}^\s + \sum_{\ell=n+2}^{N} \varepsilon_\ell , \quad 0\le n\le N-2.
\end{equation}
In particular, we get $\varepsilon_n \le 2^{N-n}\varepsilon^\s_{\max}$ for $1\le n\le N$.
\end{proposition}

The estimate provided by Proposition~\ref{prop_errorL2} is both encouraging and worrying. On the positive side, it confirms as one may expect that the approximation error can be made arbitrarily small by taking sufficiently large subvector spaces $\ch^{pr}_n$, $1\le n\le N$. On the other hand, the error might increase exponentially with the number of exercise dates, which would make the method inefficient (say for $N \ge 7$ to fix ideas). Fortunately, this is not the case on our numerical experiments. This can be heuristically explained at least by two reasons. First, in our error analysis, the projection errors add up while we may hope in practice some averaging effect and compensation between them. Second, in mathematical finance, we are not exactly interested in the $L^2$ error between $M$ and $M^\s$ in practice, but rather on the price and hedge provided by~$M$. This $L^2$ error gives an upper bound of the pricing and hedging error as pointed out by Remark~\ref{Rk_hedging}, but it may happen that $V_0=\E[\max_{0\le n \le N} \{Z_n-M_n\}]$ is much closer to~$U_0$. A more refined analysis on the approximation error is left for further research.

\begin{proof}
For $n=N$, we have proved before stating Proposition~\ref{prop_errorL2} that $\Delta M^{\s,pr}_{N}$ is the unique solution of~\eqref{eq:inf_sq}. 

Now, we consider $n\le N-1$ and assume that $\Delta M_{n+1},\dots,\Delta M_N$ are uniquely defined. To determine $\Delta M_n$, we first consider the unconstrained minimisation problem:
$$\inf_{\widetilde{\Delta M}_{n+1} \in \ch^2_{n+1}}\E\left[ \left( Z_{n} + \widetilde{\Delta M}_{n+1}  - \max_{n + 1\le j \le N} \left\{Z_{j} - \sum_{i=n+2}^{j} \Delta M^\s_i\right\}\right)^2\right].$$
Since $\widetilde{\Delta M}_{n+1}\in \cf_{n+1}$, we note that $\widetilde{\Delta M}_{n+1}$ is also the infimum of 
$$\E\left[ \left(Z_{n} + \widetilde{\Delta M}_{n+1}  - \E\left[ \max_{n + 1\le j \le N} \left\{Z_{j} - \sum_{i=n+2}^{j} \Delta M_i\right\} \bigg| \cf_{n+1} \right]\right)^2 \right],$$
and since $Z_n \in \cf_n$, we get 
$$\widetilde{\Delta M}_{n+1}= \E\left[ \max_{n + 1\le j \le N} \left\{ Z_{j} - \sum_{i=n+2}^{j} \Delta {M}_i \right\}\bigg|\cf_{n+1} \right]-\E\left[\max_{n + 1\le j \le N} \left\{ Z_{j} - \sum_{i=n+2}^{j} \Delta {M}_i \right\}\bigg|\cf_{n}\right]. $$
For $\Delta M_{n+1} \in \ch_{n+1}^{pr}$, we then get
\begin{align*}
  &\E\left[ \left(Z_{n} + \Delta M_{n+1}  - \E\left[ \max_{n + 1\le j \le N} \left\{Z_{j} - \sum_{i=n+2}^{j} \Delta M_i\right\} \bigg| \cf_{n+1} \right]\right)^2 \right] \\
  &= \E\left[ \left(Z_{n}  - \E\left[ \max_{n + 1\le j \le N} \left\{Z_{j} - \sum_{i=n+2}^{j} \Delta M_i\right\} \bigg| \cf_{n} \right]\right)^2 \right]+\E[(\widetilde{\Delta M}_{n+1}-\Delta M_{n+1})^2],
\end{align*}
so that $\Delta M_{n+1}$ is the projection of $\widetilde{\Delta M}_{n+1}$ on $\ch_{n+1}^{pr}$.

Now, we analyse the error between $\widetilde{\Delta M}_{n+1}$ and $\Delta M^\s_{n+1}$ To do so, we define $\Delta M^\ell_i= \mathbf{1}_{\ell < i}\Delta M_i
+\mathbf{1}_{\ell \ge  i}\Delta M^\s_i$, so that $\Delta M^N_i=\Delta M^\s_i$ for $1\le i\le N$ and $\Delta M^{n+1}_i= \Delta M_i$ for $i\ge n+2$. We write
\begin{align*}
 & \Delta {M}^\s_{n+1}-\widetilde{\Delta {M}}_{n+1}\\
 &= \E\left[ \max_{n + 1\le j \le N} \left\{ Z_{j} - \sum_{i=n+2}^{j} \Delta {M}^N_i \right\}\bigg|\cf_{n+1} \right]-\E\left[\max_{n + 1\le j \le N} \left\{ Z_{j} - \sum_{i=n+2}^{j} \Delta {M}^N_i \right\}\bigg|\cf_{n}\right]\\
  &-\E\left[ \max_{n + 1\le j \le N} \left\{ Z_{j} - \sum_{i=n+2}^{j} \Delta {M}^{n+1}_i \right\}\bigg|\cf_{n+1} \right]-\E\left[\max_{n + 1\le j \le N} \left\{ Z_{j} - \sum_{i=n+2}^{j} \Delta {M}^{n+1}_i \right\}\bigg|\cf_{n}\right]\\
  &=\sum_{\ell=n+2}^{N} \Bigg\{\E\left[ \max_{n + 1\le j \le N} \left\{ Z_{j} - \sum_{i=n+2}^{j} \Delta {M}^{\ell}_i \right\}\bigg|\cf_{n+1} \right]-\E\left[\max_{n + 1\le j \le N} \left\{ Z_{j} - \sum_{i=n+2}^{j} \Delta {M}^\ell_i \right\}\bigg|\cf_{n}\right]\\
  &-\E\left[ \max_{n + 1\le j \le N} \left\{ Z_{j} - \sum_{i=n+2}^{j} \Delta {M}^{\ell-1}_i \right\}\bigg|\cf_{n+1} \right]-\E\left[\max_{n + 1\le j \le N} \left\{ Z_{j} - \sum_{i=n+2}^{j} \Delta {M}^{\ell-1}_i \right\}\bigg|\cf_{n}\right]\Bigg\}\\
  &=\sum_{\ell=n+2}^{N}\E\left[ \max_{n + 1\le j \le N} \left\{ Z_{j} - \sum_{i=n+2}^{j} \Delta {M}^{\ell}_i \right\} - \max_{n + 1\le j \le N} \left\{ Z_{j} - \sum_{i=n+2}^{j} \Delta {M}^{\ell-1}_i \right\}  \bigg|\cf_{n+1} \right] \\
  &-\E\left[ \max_{n + 1\le j \le N} \left\{ Z_{j} - \sum_{i=n+2}^{j} \Delta {M}^{\ell}_i \right\} - \max_{n + 1\le j \le N} \left\{ Z_{j} - \sum_{i=n+2}^{j} \Delta {M}^{\ell-1}_i \right\}  \bigg|\cf_{n} \right] \\
  &=\sum_{\ell=n+2}^{N} \E[\mathcal{D}^\ell |\cf_{n+1}]-\E[\mathcal{D}^\ell |\cf_{n}],
\end{align*}
with 
$$ \mathcal{D}^\ell = \max_{n + 1\le j \le N} \left\{ Z_{j} - \sum_{i=n+2}^{j} \Delta {M}^{\ell}_i \right\} - \max_{n + 1\le j \le N} \left\{ Z_{j} - \sum_{i=n+2}^{j} \Delta {M}^{\ell-1}_i \right\}. $$
We get by the triangle inequality for the $L^2$ distance and noting that $\E[\Var(Z|\cf_{n})]\le \E[Z^2]$:
\begin{align*}
  \E[(\Delta {M}^\s_{n+1}-\widetilde{\Delta {M}}_{n+1})^2]^{1/2}&\le  \sum_{\ell=n+2}^{N} \E\left[\Var\left(\E[\mathcal{D}^\ell |\cf_{n+1}] |\cf_n \right)\right]^{1/2} \\
  & \le \sum_{\ell=n+2}^{N} \E\left[\E[\mathcal{D}^\ell |\cf_{n+1}]^2\right]^{1/2} \le \sum_{\ell=n+2}^{N} \E\left[(\mathcal{D}^\ell)^2\right]^{1/2} \le \sum_{\ell=n+2}^{N} \varepsilon_\ell,
\end{align*}
where the last inequality comes from $|\mathcal{D}^\ell|\le \sum_{i=n+2}^{N} |\Delta {M}^{\ell}_i -\Delta {M}^{\ell-1}_i|  =  |\Delta M_\ell -\Delta M^\s_\ell|$ (note that $\Delta M_i^\ell=\Delta M_i^{\ell-1}$ for $i\not=\ell$, $\Delta M_\ell^\ell=\Delta M^\s_\ell$ and $\Delta M^{\ell-1}_\ell=\Delta M_\ell$). Finally, we use the triangle inequality $\varepsilon_{n+1}\le \varepsilon^\s_{n+1}+\E[(\Delta {M}^{\s,pr}_{n+1}-\Delta {M}_{n+1})^2]^{1/2}$ and the fact that the projection on~$\ch_{n+1}$ is a contraction to get~\eqref{rec_error}.
Then, we get $\varepsilon_n\le 2^{N-n} \varepsilon_{\max}$ by a simple backward induction since $1+\sum_{\ell=n+1}^N 2^{N-\ell}=2^{N-n}$.
\end{proof}

\subsection{Convergence of the Monte Carlo estimators}

In general, the minimisation of~\eqref{eq:inf_sq} cannot be done explicitly, and we use a Monte Carlo algorithm to minimise the corresponding empirical mean. The goal of this subsection is to study the convergence of the Monte Carlo estimators associated with the solution of~\eqref{eq:inf_sq}.

First, let us make the setup precise. We assume that the subvector spaces $\ch_n$, $1\le n \le N$, are spanned by $L \in \N^*$ random variables $\Delta X_{n,\ell} \in \ch^2_n$, $1\le \ell \le L$:
$$\ch^{pr}_n=\left\{ \alpha \cdot \Delta X_{n} \ : \ \alpha \in \R^L \right\}.$$
Here, $L$ does not depend on $n$ for simplicity, and we note $\alpha \cdot \Delta X_{n}= \sum_{\ell=1}^L \alpha_\ell \Delta X_{n,\ell}$. Then, the minimisation problem~\eqref{eq:inf_sq}  can be rewritten as:
\begin{equation}\label{eq:inf_sq2}
  \inf_{\alpha \in \R^L }\E\left[ \left(Z_{n} + \alpha \cdot \Delta X_{n+1}  - \max_{n + 1\le j \le N} \left\{Z_{j} - \sum_{i=n+2}^{j} \Delta M_i\right\}\right)^2\right].
\end{equation} 
This is a classical least square optimisation: if the positive semidefinite matrix $\E[\Delta X_{n+1}\Delta X_{n+1}^T]$ is invertible, the minimum is given by
\begin{align}
   \alpha_{n+1} &= \left( \E[\Delta X_{n+1}\Delta X_{n+1}^T] \right)^{-1} \E\left[ \left( \max_{n + 1\le j \le N} \left\{Z_{j} - \sum_{i=n+2}^{j} \Delta M_i\right\} -Z_n \right) \Delta X_{n+1}\right] \notag\\
   &= \left( \E[\Delta X_{n+1}\Delta X_{n+1}^T] \right)^{-1} \E\left[ \left( \max_{n + 1\le j \le N} \left\{Z_{j} - \sum_{i=n+2}^{j} \Delta M_i\right\} \right) \Delta X_{n+1}\right], \label{def_alpha_n}
\end{align}
since $\E[\Delta X_{n+1}|\cf_n]=0$. 

Now, we assume that we have $Q$ independent paths of the underlying process $Z_n^q$ and of the martingale increments $\Delta X^q_n$, for $1\le n \le N$ and $1\le q\le Q$. Here, and through the paper, the index $q$ will denote the $q-$th sample of the Monte Carlo algorithm.  
% Then, we minimize, for $n=N-1$ to $n=0$, 
% \begin{equation*}
% \alpha^Q_{n+1}= \arg \min_{\alpha \in \R^L} \frac 1Q \sum_{q=1}^Q \left(Z^q_{n} + \alpha \cdot \Delta X^q_{n+1}  - \max_{n + 1\le j \le N} \left\{ Z^q_{j} - \sum_{i=n+2}^{j} \alpha^Q_i \cdot \Delta X^q_i \right\}\right)^2 .
% \end{equation*}
% Again, we write $\alpha \cdot \Delta X^q_{n+1}=\sum_{\ell=1}^L \alpha_\ell \Delta X^q_{n+1,\ell}$. 
If the matrix $\frac 1 Q \sum_{q=1}^Q \Delta X^q_{n+1} (\Delta X^q_{n+1})^T$ is invertible, then we define $\alpha^Q_{n+1}$ as the Monte Carlo estimator of~\eqref{def_alpha_n} by
\begin{equation}\label{def_alpha_n_Q}
  \alpha_{n+1}^Q=\left( \frac 1 Q \sum_{q=1}^Q \Delta X^q_{n+1} (\Delta X^q_{n+1})^T \right)^{-1} \frac 1 Q \sum_{q=1}^Q \max_{n + 1\le j \le N} \left\{ Z^q_{j} - \sum_{i=n+2}^{j} \alpha^Q_i \cdot \Delta X^q_i \right\} \Delta X^q_{n+1}
\end{equation}

\begin{proposition}\label{prop:cv_alpha} Let us assume that for $1\le n\le N$, the nonnegative symmetric matrix $\E[\Delta X_{n} \Delta X_{n}^T]$ is invertible. Then, for all $n \in \{1,\dots, N\}$, $\alpha^Q_n\to_{Q\to \infty} \alpha_n$, almost surely, where $\alpha_n$ is defined by~\eqref{def_alpha_n}.

 If we assume moreover that $\Delta X_i$ and $Z_{i}$ have finite moments of order~$4$, then $\left(\sqrt{Q}(\alpha_n^Q- \alpha_n)\right)_{Q \ge 1}$ is tight for all $n \in \{1,\dots, N\}$.   
\end{proposition}
\begin{proof}
We proceed by backward induction on~$n$. For $n=N$, we have 
$$\alpha_N^Q=  \left( \frac 1 Q \sum_{q=1}^Q \Delta X^q_{N} (\Delta X^q_{N})^T \right)^{-1} \frac 1 Q \sum_{q=1}^Q Z^q_{N} \Delta X^q_{N},$$
and this is a direct application of the strong law of large numbers since $\Delta X_N$ and $Z_{N}$ are square integrable. Now, let us assume that we have proven the almost sure convergence for $\alpha^Q_{n+2},\dots, \alpha^Q_N$. Then, we get from~\eqref{def_alpha_n_Q} that
\begin{align*}
  \alpha_{n+1}^Q&=\left( \frac 1 Q \sum_{q=1}^Q \Delta X^q_{n+1} (\Delta X^q_{n+1})^T \right)^{-1} \frac 1 Q \sum_{q=1}^Q \max_{n + 1\le j \le N} \left\{ Z^q_{j} - \sum_{i=n+2}^{j} \alpha_i \cdot \Delta X^q_i \right\} \Delta X^q_{n+1}\\
  &+\left( \frac 1 Q \sum_{q=1}^Q \Delta X^q_{n+1} (\Delta X^q_{n+1})^T \right)^{-1} \frac 1 Q \sum_{q=1}^Q \bigg[ \max_{n + 1\le j \le N} \left\{ Z^q_{j} - \sum_{i=n+2}^{j} \alpha^Q_i \cdot \Delta X^q_i \right\} \\
  & -\max_{n + 1\le j \le N} \left\{ Z^q_{j} - \sum_{i=n+2}^{j} \alpha_i \cdot \Delta X^q_i \right\} \bigg]\Delta X^q_{n+1}.
\end{align*}
Again, the first term converges almost surely to $\alpha_{n+1}$ by the SLLN. For the second term, we observe that
\begin{align}
&  \left\| \frac 1 Q \sum_{q=1}^Q \bigg[ \max_{n + 1\le j \le N} \left\{ Z^q_{j} - \sum_{i=n+2}^{j} \alpha^Q_i \cdot \Delta X^q_i \right\} 
 -\max_{n + 1\le j \le N} \left\{ Z^q_{j} - \sum_{i=n+2}^{j} \alpha_i \cdot \Delta X^q_i \right\} \bigg]\Delta X^q_{n+1}\right\| \notag \\
 &\le \frac 1 Q \sum_{q=1}^Q \max_{n + 1\le j \le N} \left|\sum_{i=n+2}^{j} (\alpha^Q_i-\alpha_i) \cdot \Delta X^q_i \right| \|\Delta X^q_{n+1}\|  \label{bound_norm_max}\\
 &\le  \frac 1 Q \sum_{q=1}^Q  \sum_{i=n+2}^{N} \|\alpha^Q_i-\alpha_i \| \| \Delta X^q_i\| \|\Delta X^q_{n+1}\| \le \max_{n+2\le i \le N} \|\alpha^Q_i-\alpha_i \| \frac 1 Q \sum_{q=1}^Q  \sum_{i=n+2}^{N} \| \Delta X^q_i\| \| \Delta X^q_{n+1} \|,\notag 
\end{align}
which converges to $0$ almost surely by the induction hypothesis and the SLLN.

Now, we prove the tightness and proceed by backward induction. For $n=N$, we write
\begin{align*}
  \sqrt{Q}(\alpha_N^Q - \alpha_N) &=  \left( \frac 1 Q \sum_{q=1}^Q \Delta X^q_{N} (\Delta X^q_{N})^T \right)^{-1}  \sqrt{Q}\left( \frac 1 Q \sum_{q=1}^Q Z^q_{N} \Delta X^q_{N} - \E[Z_{N} \Delta X_N] \right) \\
  &+ \sqrt{Q}\left( \left(\frac 1 Q \sum_{q=1}^Q \Delta X^q_{N} (\Delta X^q_{N})^T \right)^{-1} -\E[\Delta X_{N} (\Delta X_{N})^T]^{-1}\right)\E[Z_{N} \Delta X_N].
\end{align*}
As $Q\to \infty$, the first term converges in distribution to a multivariate normal distribution by the CLT and Slutsky's Lemma. The second term converges also to a multivariate normal distribution by the CLT and the Delta method. The sequence $\left(\sqrt{Q}(\alpha_N^Q - \alpha_N)\right)_{Q\ge 1}$ is thus tight.

Now, we consider $n+1<N$ and assume that $ \left(\sqrt{Q}(\alpha_i^Q - \alpha_i)\right)_{Q\ge 1} $ is tight for $i=n+2,\dots,N$. Then, we write 
\begin{align*}
  \sqrt{Q}(\alpha_{n+1}^Q - \alpha_{n+1}) =  &\left( \frac 1 Q \sum_{q=1}^Q \Delta X^q_{n+1} (\Delta X^q_{n+1})^T \right)^{-1}  \left(\mathcal{T}^Q_1 +\mathcal{T}^Q_2\right) \\
  &   +\mathcal{T}^Q_3 \E\left[ \max_{n + 1\le j \le N} \left\{ Z_{j} - \sum_{i=n+2}^{j} \alpha_i \cdot \Delta X_i \right\}  \Delta X_{n+1}\right],
\end{align*}
with
{\footnotesize
\begin{align*}
  \mathcal{T}^Q_1&=\sqrt{Q}\left( \frac 1 Q \sum_{q=1}^Q \left( \max_{n + 1\le j \le N} \left\{ Z^q_{j} - \sum_{i=n+2}^{j} \alpha^Q_i \cdot \Delta X^q_i \right\}  -\max_{n + 1\le j \le N} \left\{ Z^q_{j} - \sum_{i=n+2}^{j} \alpha_i \cdot \Delta X^q_i \right\} \right)\Delta X^q_{n+1}    \right), \\
  \mathcal{T}^Q_2&=\sqrt{Q}\left( \frac 1 Q \sum_{q=1}^Q \max_{n + 1\le j \le N} \left\{ Z^q_{j} - \sum_{i=n+2}^{j} \alpha_i \cdot \Delta X^q_i \right\} \Delta X^q_{n+1}   - \E\left[ \max_{n + 1\le j \le N} \left\{ Z_{j} - \sum_{i=n+2}^{j} \alpha_i \cdot \Delta X_i \right\}  \Delta X_{n+1}\right]  \right), \\
  \mathcal{T}^Q_3&=\sqrt{Q} \left( \left(\frac 1 Q \sum_{q=1}^Q \Delta X^q_{n+1} (\Delta X^q_{n+1})^T \right)^{-1} -\E[\Delta X_{n+1} (\Delta X_{n+1})^T]^{-1}\right).
\end{align*}
}

As $Q\to \infty$, the terms $\mathcal{T}^Q_2$ and $\mathcal{T}^Q_3$ converge in distribution to a multivariate normal distribution by the CLT. For the first term, we  use the same arguments as in~\eqref{bound_norm_max}  and get
\begin{align*}
  & \left\|    \mathcal{T}^Q_1 \right\| \le \max_{n+2\le i \le N} \|\sqrt{Q}(\alpha^Q_i-\alpha_i) \| \frac 1 Q \sum_{q=1}^Q  \sum_{i=n+2}^{N} \| \Delta X^q_i\| \| \Delta X^q_{n+1} \|. 
\end{align*}
By the SLLN, $\frac 1 Q \sum_{q=1}^Q  \sum_{i=n+2}^{N} \| \Delta X^q_i\| \| \Delta X^q_{n+1} \|$ converge almost surely and $(\mathcal{T}^Q_1)_{Q\ge 1}$ is tight by using the induction hypothesis. Therefore, $(\mathcal{T}^Q_1+\mathcal{T}^Q_2)_{Q\ge 1}$  is also tight. Since  $\left( \frac 1 Q \sum_{q=1}^Q \Delta X^q_{n+1} (\Delta X^q_{n+1})^T \right)^{-1}$ converges almost surely to $\left(\E[(\Delta X_{n+1} (\Delta X_{n+1})^T)]\right)^{-1}$, we get that the sequence $\left(\sqrt{Q}(\alpha_{n+1}^Q - \alpha_{n+1})\right)_{Q \ge 1}$ is  tight.
\end{proof}

\begin{corollary}\label{cor_price}
  We make the same assumption as in Proposition~\ref{prop:cv_alpha}
  Let $U_0^Q= \frac 1 Q \sum_{q=1}^Q \max_{0\le j \le N} \left\{ Z^q_{j} - \sum_{i=1}^j \alpha^Q_i \cdot \Delta X^q_i \right\}$ be the estimated Monte-Carlo dual price. Then, 
  $$  U_0^Q \underset{Q\to \infty} \to \E \left[ \max_{0\le j \le N} \left\{ Z_{j} - \sum_{i=1}^j \alpha_i \cdot \Delta X_i \right\} \right],$$
  and $\left(\sqrt{Q}\left( U_0^Q - \E \left[ \max_{0\le j \le N} \left\{ Z_{j} - \sum_{i=1}^j \alpha_i \cdot \Delta X_i \right\} \right] \right)\right)_{Q\ge 1}$ is a tight sequence. 
\end{corollary}
\begin{proof}
  We have 
  $$ U_0^Q= \frac 1 Q \sum_{q=1}^Q \max_{0\le j \le N} \left\{ Z^q_{j} - \sum_{i=1}^j \alpha_i \cdot \Delta X^q_i \right\} + \Delta^Q,  $$
  with $\Delta^Q=\frac 1 Q \sum_{q=1}^Q  \max_{0 \le j \le N} \left\{ Z^q_{j} - \sum_{i=1}^j \alpha^Q_i \cdot \Delta X^q_i \right\} - \max_{0 \le j \le N} \left\{ Z^q_{j} - \sum_{i=1}^j \alpha_i \cdot \Delta X^q_i \right\}$. We have $|\Delta^Q|\le \sum_{i=1}^N |\alpha_i^Q-\alpha_i| \frac 1Q \sum_{q=1}^Q |\Delta X_i^q| $. From Proposition~\ref{prop:cv_alpha}, We get that $\Delta^Q \to 0$ a.s. and that $(\sqrt{Q} \Delta^Q)_{Q \ge 1}$ is tight since $\frac 1Q \sum_{q=1}^Q |\Delta X_i^q| $ converges almost surely by the LLN. The CLT gives the tightness of $\sqrt{Q}\left( \frac 1 Q \sum_{q=1}^Q \max_{0\le j \le N} \left\{ Z^q_{j} - \sum_{i=1}^j \alpha_i \cdot \Delta X^q_i \right\} -\E \left[ \max_{0\le j \le N} \left\{ Z_{j} - \sum_{i=1}^j \alpha_i \cdot \Delta X_i \right\} \right] \right)$ and thus the claim. 
\end{proof}

\section{Implementation in a financial framework}\label{Sec_Implementation}

We present the application framework and consider a market with $d$ assets $(S^k_t,t\ge 0)$, $k\in\{1,\dots, d\}$. We denote by $(\mathcal{G}_t,t\ge 0)$ the usual filtration generated by these assets. We consider for simplicity a deterministic interest rate~$r$, and assume that the discounted assets $(\tilde{S}^k_t,t\ge 0)$ with $\tilde{S}_t^k=e^{-rt}S_t^k$ are square integrable $\mathcal{G}_t$-martingales. 
We consider a time horizon~$T>0$ and a Bermudan option with regular exercising dates
$$T_i=\frac{iT}{N}, \ i=0,\dots,N.$$ We take $\cf_i=\mathcal{G}_{T_i}$ to fit with the framework of the previous sections. 
Since perfect hedging is only attainable in continuous time when a martingale representation theorem holds, we suggest to use a sub-grid to parametrise the martingale increments. Each interval $[T_i, T_{i+1}]$ for $0 \le i \le  N-1 $ is split into $\bar{N}$ regular sub-intervals, and we set 
\begin{equation}\label{def_subticks} t_{i,j}=T_i+\frac{j}{\bar{N}} \frac{T}{N}, \text{ for }0 \le j \le \bar{N}.\end{equation}

We consider a family of functions $u_p:\R^d \to \R$ for $p\in \{1,\dots,\bar{P}\}$ and a family of discounted assets $(\ca^k)_{1\le k\le \bar{d}}$. When the hedge is only made with the underlying assets, we take $\bar{d}=d$ and  $\ca^k=\tilde{S}^k$ but we may also include European option claims in the hedge and take $\bar{d}>d$ in this case. Then, we define the  following elementary martingale increments:
\begin{equation}\label{def_mg_incr}
  X^{{p},k}_{t_{i,j}} - X^{{p},k}_{t_{i,j-1}} =  u^{p}_{i,j-1} (S_{t_{i,j-1}}) (\ca^k_{t_{i,j}}- \ca^k_{t_{i,j-1}}),
\end{equation}
for $1\le p\le \bar{P}$ and $1\le k\le \bar{d}$. Thus,  $L=\bar{N} \times \bar{P} \times \bar{d}$  is the number of martingale increments between two exercising dates that span $\ch_i^{pr}$. It is also the dimension of $\ch_i^{pr}$ if this is a free family. In this paper, we consider for $u^p_{i,j}$ either polynomial functions or indicator functions (local basis). This is further described in Appendix~\ref{App_basis}.
Finally, we decompose the martingale increments $\Delta M_{i+1}$, $0\le i \le N-1$ as follows
\begin{align}\label{def_DeltaMi}
  \Delta M_{i+1} =  \sum_{j=1}^{\bar{N}}\sum_{{{ p},k}} \alpha_{i,j}^{{ p},k}  (X^{{ p},k}_{t_{i, j}} - X^{{ p},k}_{t_{i, j-1}}),
\end{align}
and note $\alpha_i$ the vector $(\alpha_{i,j}^{{ p},k})_{j,p,k} \in \R^L$.

At stage $i$, we optimize over all the coefficients $\alpha^{{ p},k}_{i,j}$ for $1 \le j \le \bar{N}$, ${ p}\in \{1,\dots,\bar{P}\}$ and $k\in\{1,\dots, \bar{d}\}$. 
Namely, we specify~\eqref{def_alpha_n} in this framework. We first focus on the use of subticks. By using the martingale property of $X^{{ p},k}$, we get for $i$ and $j$ given, 
\begin{align} \label{linear_system_alpha}\sum_{{ p}',k'}  &\alpha_{i,j}^{{ p}',k'} \E[(X^{{ p},k}_{t_{i,j}}-X^{{ p},k}_{t_{i,j-1}})(X^{{ p}',k'}_{t_{i,j}}-X^{{ p}',k'}_{t_{i,j-1}})]\\=&\E\left[(X^{{ p},k}_{t_{i,j}}-X^{{ p},k}_{t_{i,j-1}}) \left(\max_{i + 1\le m \le N} \left\{ Z_{m} - \sum_{\ell=i+2}^{m} \alpha_\ell \cdot \Delta X_\ell \right\}  \right)\right], \notag
\end{align}
for all $p$ and $k$. 
Thus, the use of sub-intervals has a linear computational cost: instead of solving a linear system of size $L=\bar{N}\times \bar{P}\times \bar{d}$, we solve $\bar{N}$ linear systems of size $\bar{P}\times \bar{d}$. Besides, the linear systems can be solved independently for the different values of $1 \le j \le \bar{N}$. In fact, the use of subticks makes the problem similar to a European option that pays $\max_{i + 1\le m \le N} \left\{ Z_{m} - \sum_{\ell=i+2}^{m} \alpha_\ell \cdot \Delta X_\ell \right\}$ at time $t_{i,\bar{N}}=T_{i+1}$, and we are computing the mean-variance hedging of this option on $[T_i,T_{i+1}]$ on the time grid $t_{i,j}$. The decomposition in $\bar{N}$ independent minimisation problems is pointed out in Remark~\ref{Rk_Europ} 
{
\begin{remark}[Computational complexity] In practice, we consider a Monte-Carlo estimator of~\eqref{linear_system_alpha} with $Q$ samples and we have to solve $N\times \bar{N}$ linear systems of dimension $\bar{P}\times \bar{d}$. Thus, the overall complexity is of order $O(Q N \bar{N}(\bar{P}\bar{d})^3 )$, which can even be $O(Q N \bar{N}\bar{P}\bar{d}^3 )$ when choosing a local regression basis.  In contrast, the global minimisation proposed in~\cite{DFM} requires to solve a Linear Program with $Q+N\times \bar{N}\times \bar{P} \times \bar{d}$ variables and $Q\times (N+1)$ constraints and is therefore far most computationally demanding, see e.g.~\cite{CLZ} for recent results on the LP complexity. 
\end{remark}
}

\section{Numerical results}\label{Sec_Num}

In all our numerical experiments, we consider a Black Scholes model, defined on a risk neutral probability space $(\Omega,\overline{\cf},\P)$  by
\begin{equation}\label{eq:BS_model}
  dS^k_t = S^k_t \left((r - \delta^k ) dt + \sigma^k dW^k_t \right),
\end{equation}
for $k=1,\dots,d$ where for every $k$, $\sigma^k$ and $\delta^k$ are respectively the volatility and the dividend rate of asset~$k$, $r$ is the constant interest rate and $W$ is a $d$-dimensional Brownian motion such that $d\langle W^k, W^{k'} \rangle_t = \rho \ind{k \ne k'} dt + \ind{k = k'} dt$, with $\rho \in \left[-\frac{1}{d-1},1\right]$. We set $\tilde{S}^k_t=e^{(\delta^k-r)t}S^k_t$, which is a martingale representing the discounted value of asset~$k$ when $\delta^k=0$. More generally when $\delta_k \in \R$, this is the discounted value of a self financing portfolio that only invests in asset~$k$ (dividends are instantly reinvested in asset~$k$) with initial value $S_0^k$. 

In all cases, we are interested in Bermudan option prices with $N$ regular exercising dates on $[0,T]$ and discounted payoff $Z_n=e^{-r\frac{nT}N}\Psi(S_{\frac{nT}N})$ at time $n \in \{0,\dots, N\}$. We use Corollary~\ref{cor_price} to approximate the dual price by
$$ U_0^Q= \frac 1 Q \sum_{q=1}^Q \max_{0\le j \le N} \left\{ Z^q_{j} - \sum_{i=1}^j \alpha^Q_i \cdot \Delta X^q_i \right\}.$$ 
Note that the optimisation procedure leading to the computation of the coefficients $\alpha^Q$ and the Monte Carlo estimators use the same samples. Because of overfitting, $U_0^Q$ can significantly underestimate $\E \left[ \max_{0\le j \le N} \left\{ Z_{j} - \sum_{i=1}^j \alpha_i \cdot \Delta X_i \right\} \right]$ when $Q$ is not sufficiently large, compared to the number of parameters to estimate. Thus, it can be wise to compute a second approximation of the dual price using new simulations:
\begin{equation}
  \hat{U}_0^Q= \frac 1 Q \sum_{q=1}^Q \max_{0\le j \le N} \left\{ \hat{Z}^q_{j} - \sum_{i=1}^j \alpha^Q_i \cdot \Delta \hat{X}^q_i \right\},
\end{equation} 
where $(\hat{Z}^q,\Delta \hat{X}^q )_{1\le q \le Q}$ is independent from the sample $(Z^q,\Delta X^q )_{1\le q \le Q}$ used to compute $\alpha^Q$. Note that $\hat{U}_0^Q$ has a nonnegative bias with respect to~$U_0$ since it is the Monte Carlo estimator of $\E\left[ \max_{0 \le j \le N} \{Z_{j} - M_{j}\}\right]$ for the martingale $M$ associated with $\alpha^Q$. In practice, when the difference between $U_0^Q$ and $\hat{U}_0^Q$ is relatively small, this means that there nos  overfitting, and that the value of $Q$ is appropriate for the number of parameters to estimate. 

Let us recall here that $\hat{U}_0^Q$ is the true initial value of the hedging strategy computed by our algorithm. Contrary to primal methods, the prices that we are reporting in the tables below correspond to effective discrete time hedging strategies. Besides, we show on the different examples the P\&L given by the hedging strategy thus obtained. We consider the point of view of an option seller who hedges against the best exercising policy~$\hat{\tau}^\s$, which we compute with the Longstaff Schwartz algorithm. We draw the empirical distribution of $\left(\hat{U}_0^Q+ \sum_{i=1}^{\hat{\tau}^\s} \alpha_i^Q \cdot \Delta \hat{X}^q_i -\hat{Z}^q_{\hat{\tau}^\s}\right)_{1\le q \le Q}$ on different examples. The more peaked is the distribution, the better is the hedge. Thus, we can analyse empirically the interest of including European options in the hedging portfolio or of increasing the rebalancing frequency ($\bar{N}$).

\subsection{One dimensional options}\label{subsec_num1d}

We start our discussion by considering a Bermudan put option with payoff function
$$ \Psi(S) = (K - S)_+.$$
We consider two types of hedging portfolio: one using only the underlying asset ($\ca^1_t=\tilde{S}_t$ in~\eqref{def_mg_incr}) and a second one using both the underlying asset and the European put option with the same strike ($\ca^1_t=\tilde{S}_t$ and $\ca^2_t=e^{-rT }\E[\Psi(S_T)|\cf_t]$ in~\eqref{def_mg_incr}). The use of the option is indicated by a boolean in Table~\ref{Table:Put_local} in the column ``Vanilla''. Concerning the regression basis, we present here the results obtained by the local basis presented in Appendix~\ref{App_basis}. The results obtained with a polynomial basis are quite similar and we have decided not to include them. The use of polynomial functions is more relevant in higher dimension to reduce the size of the projection basis.

\begin{table}[htbp!]
  \centering\begin{tabular}{cccccc}
 \hline
    $Q$ & $\bar{N}$ & $P$ & Vanilla\phantom{$\Big|$} &  $U_0^Q$ & $\hat{U}_0^Q$ \\
    \hline
    50000 & 1 & 1 & True & 9.91 & 9.91 \\
    100000 & 1 & 50 & True & 9.89 & 9.91 \\
    100000 & 1 & 50 & False & 10.32 & 10.33 \\
    100000 & 5 & 50 & False & 9.99 & 10.08 \\
    100000 & 10 & 100 & False & 9.82 & 10.19 \\
    500000 & 10 & 100 & False & 9.95 & 10.02 \\
    2000000 & 10 & 50 & False & 9.98 & 9.98 \\
    2000000 & 20 & 50 & False & 9.94 & 9.96 \\
    \hline
    \end{tabular}
  \caption{Prices for a put option using a basis of $P$ local functions with $K = S_0 = 100$, $T = 0.5$, $r=0.06$, $\sigma=0.4$ and $N=10$ exercising dates. LS price with a polynomial approximation of order 6: $9.90$. \label{Table:Put_local}}
\end{table}

\begin{figure}[h!]
  \begin{subfigure}{0.33\textwidth}
    \includegraphics[width=\textwidth]{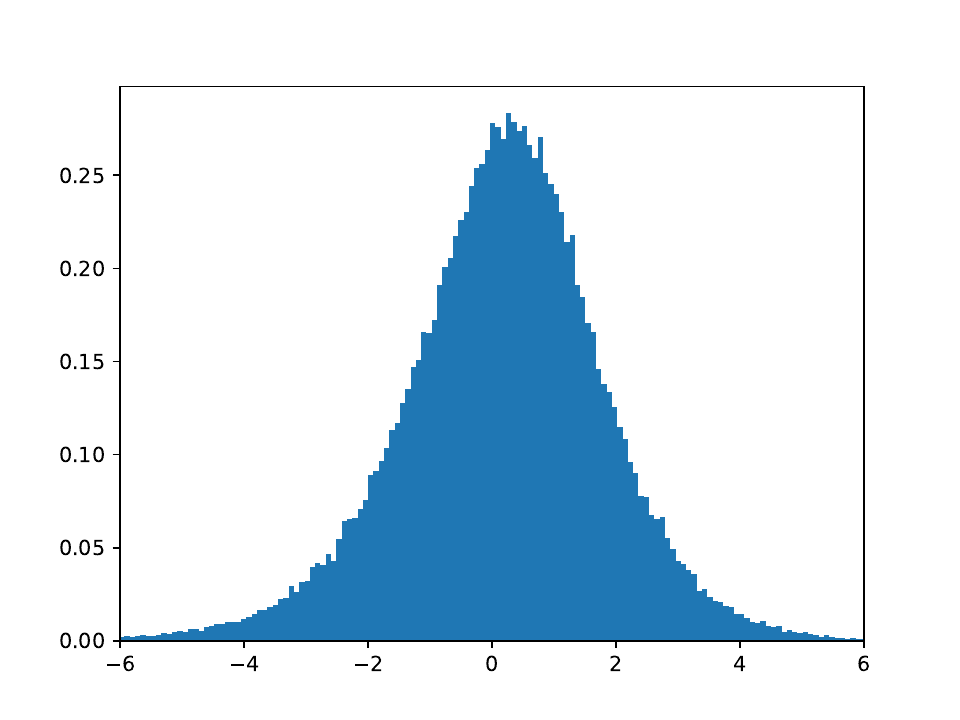}\caption{$\bar{N}=5$, $P=50$, $Q=10^5$.}\label{fig:Put_local_a}
  \end{subfigure}
  \begin{subfigure}{0.33\textwidth}
    \includegraphics[width=\textwidth]{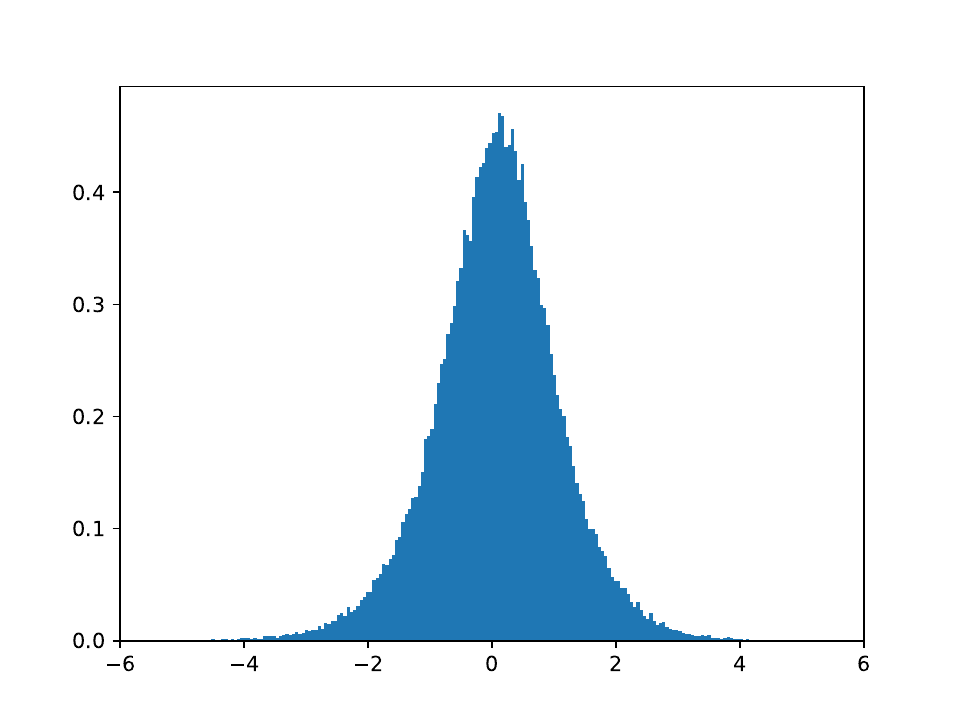}\caption{$\bar{N}=10$, $P=50$, $Q=2\times 10^6$.}\label{fig:Put_local_b}
  \end{subfigure}
   \begin{subfigure}{0.33\textwidth}
    \includegraphics[width=\textwidth]{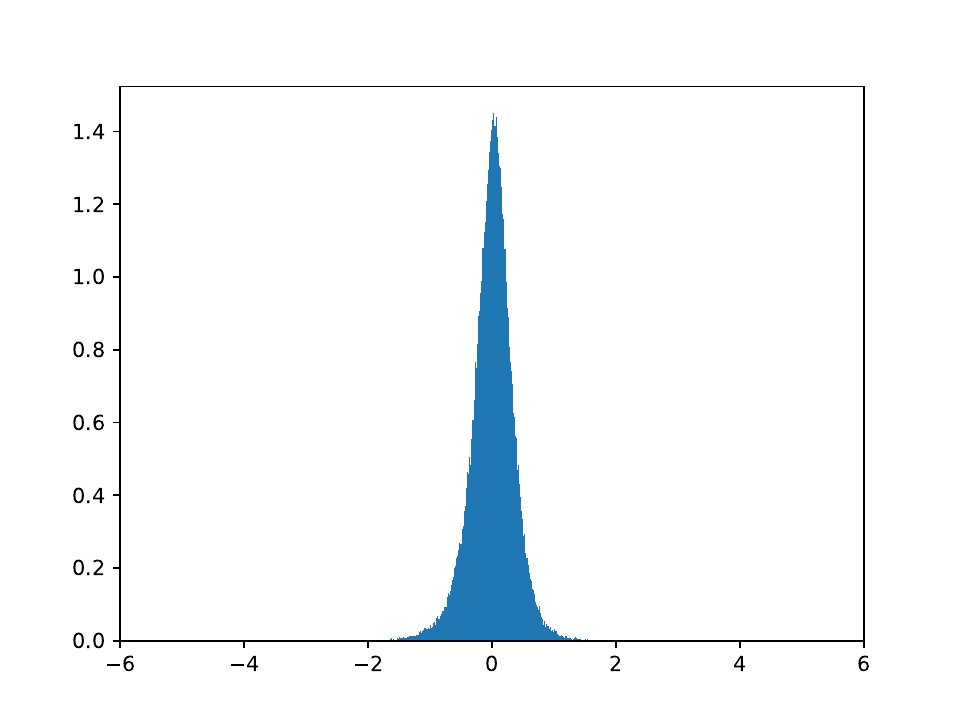}\caption{$\bar{N}=1$, $P=50$, $Q=10^5$.}\label{fig:Put_local_c}
\end{subfigure} \\
\begin{center}\begin{subfigure}{0.33\textwidth}
  \includegraphics[width=\textwidth]{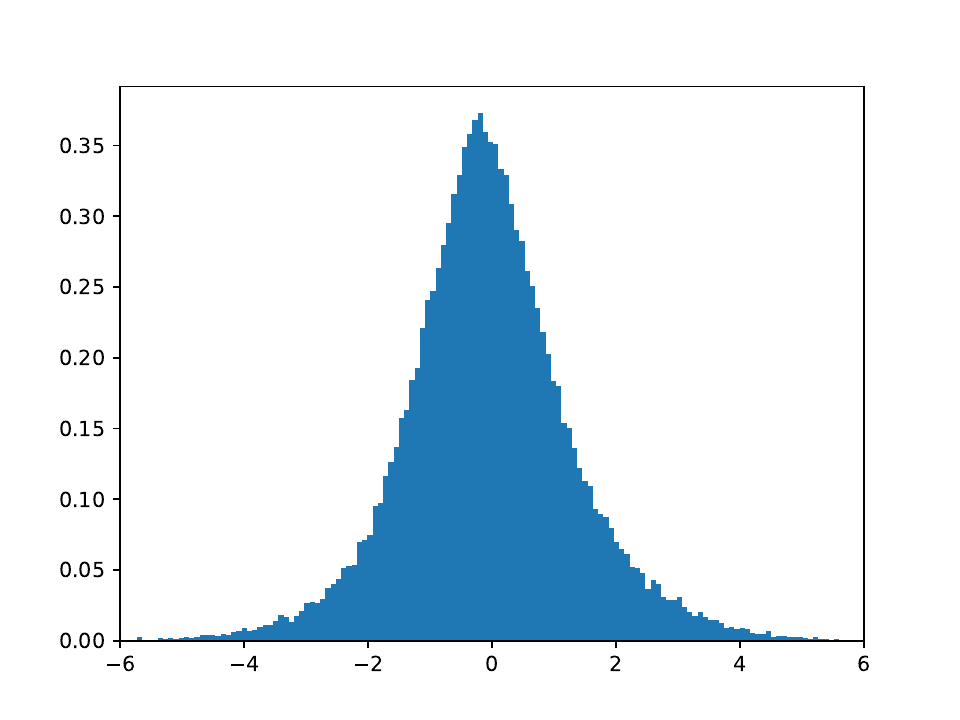}\caption{$\bar{N}=5$, $P=50$, $Q=10^5$, delta hedging (CRR).}\label{fig:Put_local_CRR}
\end{subfigure}
\begin{subfigure}{0.33\textwidth}
  \includegraphics[width=\textwidth]{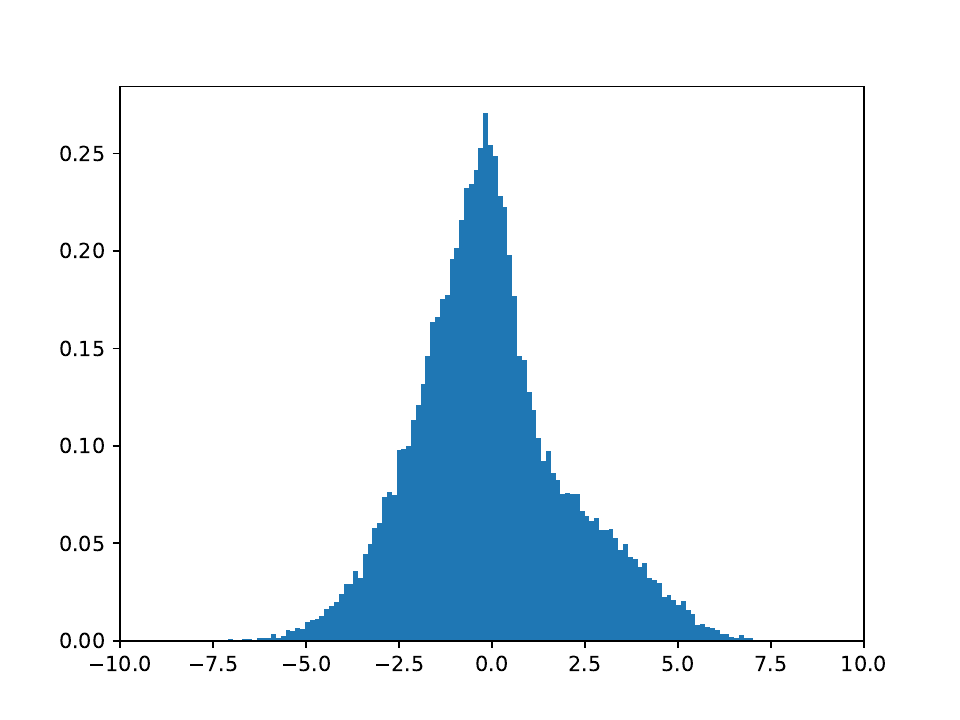}\caption{$\bar{N}=5$, $P=50$, $Q=10^5$, delta hedging (Wang Caflisch).}\label{fig:Put_local_WC}
\end{subfigure}
\end{center}
  \caption{P\&L histograms of the hedging strategy for the Bermudan Put option of Table~\ref{Table:Put_local} for the stock only strategy ((\ref{fig:Put_local_a}) and (\ref{fig:Put_local_b}))  and the strategy using extra European options (\ref{fig:Put_local_c}). {Figures~(\ref{fig:Put_local_CRR}) and~(\ref{fig:Put_local_WC}) show the P\&L with the delta hedging strategy computed explicitly or by Monte-Carlo in the same framework as Figure~(\ref{fig:Put_local_a}).}} \label{fig:Put_local}
\end{figure}

First, we comment the results with the stock only portfolio. As we may expect, the rebalancing frequency plays a crucial role. The higher is the frequency, the closer is the estimated dual price to the primal price given by the Longstaff and Schwartz algorithm. 
We observe on this example the effect of overfitting: for $\bar{N}=10$ and $P=100$, the total number of parameters to estimate is $N\times\bar{N}\times P=10000$, and $Q=100000$ is not sufficient for a suitable estimation, as shown by the difference between $U_0^Q$ and $\hat{U}_0^Q$. Increasing $Q$ to $500000$ reduces this gap. With $\bar{N}=20$, $P=50$ and $Q=2\times 10^6$, we get a price which is $0.06$ above the Longstaff-Schwartz price, which is satisfactory: the same error magnitude appears also when hedging European options with discrete time strategies. 

We now discuss the P\&L obtained in Figure~\ref{fig:Put_local}. We observe the importance of the rebalancing frequency for a stock only portfolio, and get a variance of 2.73 when $\bar{N}=5$ (Figure~\ref{fig:Put_local_a}) and $1.05$ when $\bar{N}=10$ (Figure~\ref{fig:Put_local_b}). {We have also compared the hedging strategy given by our algorithm with the classical delta hedging strategy. In Figure~\ref{fig:Put_local_CRR}, we show the P\&L distribution obtained with the delta hedging strategy when the delta is computed explicitly with the Cox-Ross-Rubinstein (CRR) approximation of the Black-Scholes model (see~\cite{Hull} and~\cite{PV}). Compared to Figure~\ref{fig:Put_local_a}, the distributions are very similar with a slightly sharper hedge for the explicit delta hedging. However, the delta cannot be computed explicitly in general and one has to use approximations method such as the Monte Carlo method proposed by~\cite{WC}. In Figure~\ref{fig:Put_local_WC}, we have plotted the P\&L of the delta hedging obtained with this method when the delta strategy is computed
once and for all at the initial time, with a polynomial regression of order 6. The distribution is again quite similar to the one of Figure~\ref{fig:Put_local_a}, but it is slightly less peaked around zero. These comparisons show that our algorithm produces a sound hedge with respect to classical methods.}

When adding the European put option in the hedging portfolio, our algorithm is very close to the primal price with already $\bar{N}=1$ and $P=1$, as expected from the conclusions of~\cite{rogers-02}. To be precise, let us consider the following martingale increments between two exercising dates
$$\Delta \ca^2_i=e^{-rT_i}P_{BS}(T_i,S_{T_i})- e^{-rT_{i-1}}P_{BS}(T_{i-1},S_{T_{i-1}}),$$ 
where $P_{BS}(t,x)$ denotes the Black-Scholes price of a European put option (with maturity $T$ and strike $K$) at time $t\in [0,T]$ and spot value $S_t=x$.  
Instead of minimizing $\E[\max_{0\le j\le N}\left(Z_j-\sum_{i=1}^j \alpha_i \Delta \ca^2_i \right)]$, it is suggested in~\cite{rogers-02} to simply minimize numerically
\begin{equation}\label{Rogers_algo}\E\left[\max_{0\le j\le N}\left(Z_j- \alpha \sum_{i=1}^j \Delta \ca^2_i \right)\right],
\end{equation}
i.e. the coefficient $\alpha$ is constant along the time. Using this algorithm with $Q=10^6$, we obtain that the price of the Bermudan option considered in Table~\ref{Table:Put_local} is 9.96, which is already very close to the Longstaff-Schwartz price.

On some other examples, it may be less obvious to find a European option that gives a good proxy of the Martingale $M^\s$, which is crucial in the  algorithm developed by~\cite{rogers-02}. To illustrate this, we consider a butterfly option with payoff function $$\Psi(S)=2\left(\frac{K_1+K_2}2-S\right)_+ - \left(K_1-S\right)_+ - \left(K_2-S\right)_+.$$
A natural choice is to take the European option with the same payoff function as the Bermudan, which lead us to consider:
$$  \ca^2= e^{-r T_i}BTF_{BS}(T_i,S_{T_i}) -e^{-r T_{i-1}}BTF_{BS}(T_{i-1},S_{T_{i-1}}), $$ 
where $BTF_{BS}(t,x)$ denotes the Black-Scholes price of a European butterfly option with payoff $\Psi$ at time~$t$. Then, applying the same Monte-Carlo algorithm to minimize~\eqref{Rogers_algo} with $Q=10^6$, $S_0=95$, $K_1=90$, $K_2=110$, $T=0.5$, $r=0.06$, $\sigma=0.4$ and $N=10$, we get a price equal to 6.49 while the Longstaff Schwartz price is equal to 5.65. % (see Table~\ref{tab:butterfly_loc}).
 Using the European put option with strike~$\frac{K_1+K_2}2$ instead of the European butterfly option in this algorithm leads to a dual price equal to 7.04, which is worse. This shows that Rogers' algorithm, which is very efficient from a numerical point of view, requires to have at hand a martingale that is already quite a good approximation of $M^\s$. When such martingale is not available, one has for the dual approach to optimize over a larger family of martingales as done for example in~\cite{BH,DFM,Lelong} and as we propose in this paper. 

\begin{table}[htbp!]
  \centering\begin{tabular}{cccccc}
    \hline
    $Q$ & $\bar{N}$ & $P$ & Vanilla\phantom{$\Big|$} & $U_0^Q$ & $\hat{U}_0^Q$ \\
    \hline
    50000 & 1 & 50 & False & 6.54 & 6.54 \\
    50000 & 1 & 50 & True & 6.25 & 6.28 \\
    100000 & 5 & 50 & False & 6.11 & 6.12 \\
    100000 & 5 & 50 & True & 5.91 & 5.95 \\
    100000 & 10 & 50 & False & 5.97 & 6.00 \\
    100000 & 10 & 50 & True & 5.79 & 5.87 \\
    500000 & 10 & 100 & False & 5.97 & 5.98 \\
    500000 & 10 & 100 & True & 5.79 & 5.82 \\
    500000 & 20 & 50 & False & 5.86 & 5.87 \\
    500000 & 20 & 50 & True & 5.71 & 5.74 \\
    \hline
    \end{tabular}
    \caption{Prices for a butterfly option with parameters using a basis of~$P$ local functions. The Longstaff-Schwartz algorithm with order $5$ polynomials gives a price of $5.65$. \label{tab:butterfly_loc}}
\end{table}

For the same Bermudan Butterfly option, Table~\ref{tab:butterfly_loc} shows the result  of our algorithm when using a hedging portfolio using only the asset ($\ca^1_t=\tilde{S}_t$) or using the asset and a European call option ($\ca^1_t=\tilde{S}_t$ and $\ca^2_t=e^{-rT }\E[(S_T-\frac{K_1+K_2}2)_+ |\cf_t]$ in~\eqref{def_mg_incr}). We see again a clear advantage to include the European option in the hedging strategy and to increase the rebalancing frequency. Thus, our algorithm does not only give a hedging strategy, it can also be used by practitioners to find a good trade-off between the price and the cost of extra monitoring. Besides, we show in  Figure~\ref{fig:butterfly_loc}  that including European options does not only improve the selling price but gives a lower variance of the P\&L distribution.  

\begin{figure}[h]
  \begin{center}
    \includegraphics[width=0.4\textwidth]{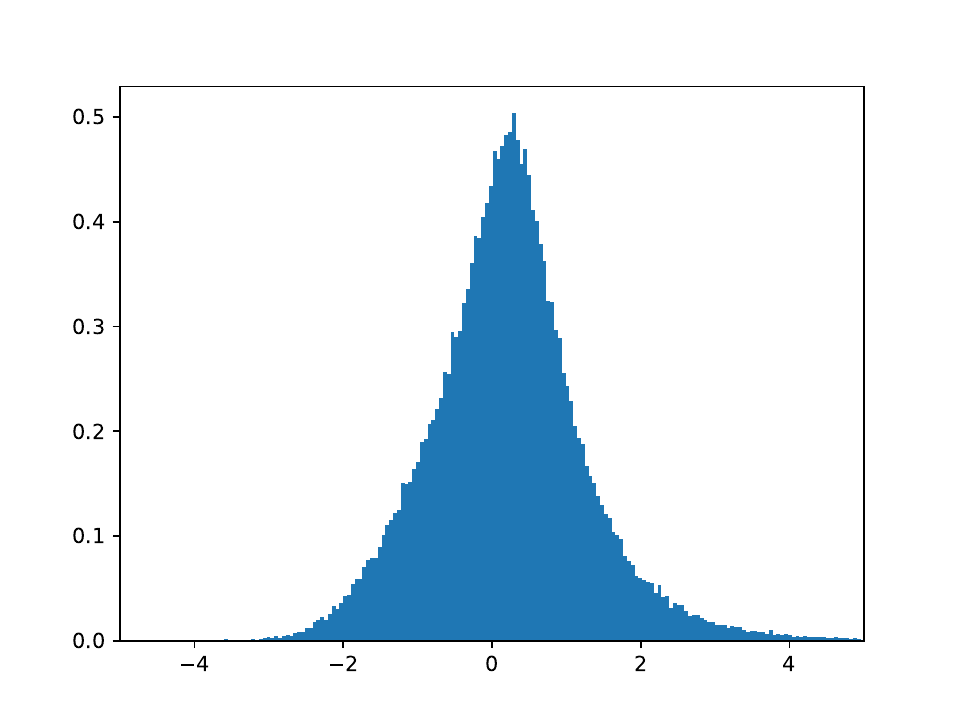}
    \includegraphics[width=0.4\textwidth]{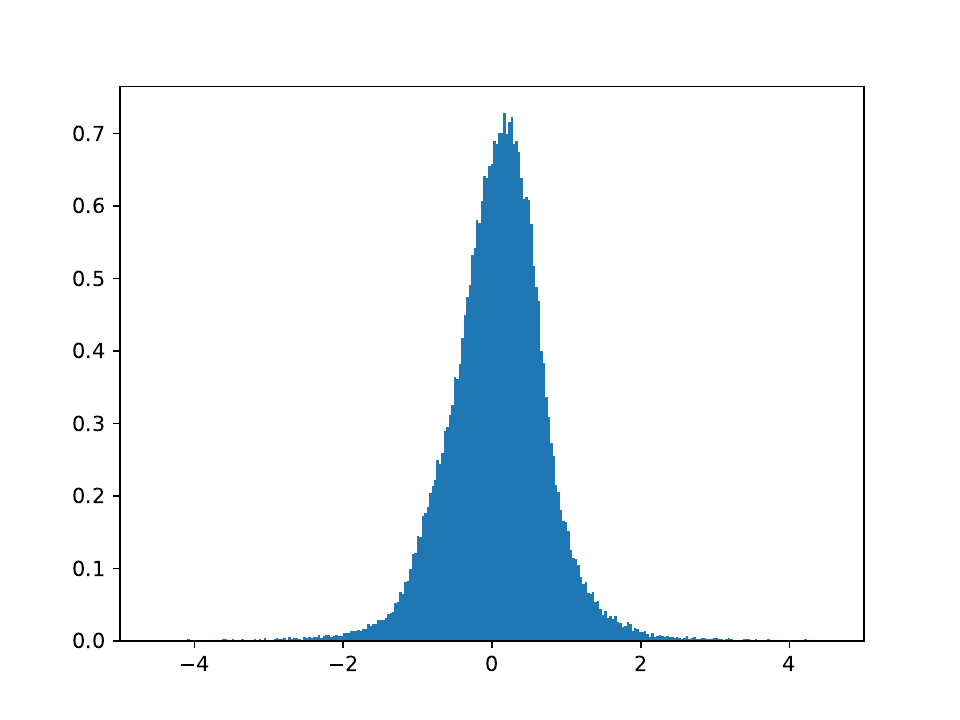}
  \end{center}
  \caption{P\&L histograms of the hedging strategy for the Bermudan Butterfly option of Table~\ref{tab:butterfly_loc} obtained with $\bar{N}=20$, $P=50$, $Q=5\times 10^5$ for the stock only strategy (left) and the strategy using extra European options (right). } \label{fig:butterfly_loc}
\end{figure}

\subsection{Bermudan options on many assets}

Now, we turn to experiments on multidimensional options. Unlike the one dimensional case, we now present two different bases for the regression: the piecewise constant local basis and the basis of polynomial functions. In fact, the local basis suffers from the curse of dimensionality and alternative choices can be relevant. In all these experiments, we either consider a portfolio using only the underlying assets (i.e. $\bar{d}=d$ and $\ca^k_t=\tilde{S}^k_t$ for $k=1,\dots,d$ in~\eqref{def_mg_incr}) or a portfolio using  the underlying assets and their corresponding at-the-money  
European call options (i.e. $\bar{d}=2d$, $\ca^k_t=\tilde{S}^k_t$ and $\ca^{d+k}_t=e^{-rT }\E[(S^k_T-S^k_0)_+ |\cf_t]$ for $k=1,\dots,d$ in~\eqref{def_mg_incr}). This choice is indicated in the tables  by a boolean in the ``Vanilla'' column. Of course, other choices of European options are possible. It may exist more clever choices, depending on the  considered Bermudan option. Here, we adopt a systematic approach and leave the investigation of finding the more appropriate hedging instruments for future research.

\subsubsection{Max-call option}

We consider a call option on the maximum of a basket of assets with payoff:
\[ \Psi(S)= \left(K - \max_{1 \le k \le d} S^k\right)_+.\]
We have tested our algorithm on the example presented in~\cite[Table 2]{AB}. Results are reported in Table~\ref{tab:max-call_loc} when using the local regression basis and in Table~\ref{tab:max-call_pol} when using polynomial functions.

\begin{table}[h!]
  \centering\begin{tabular}{cccccc}
    \hline
    $Q$ & $\bar{N}$ & $P$ & Vanilla\phantom{$\Big|$} & $U_0^Q$ & $\hat{U}_0^Q$ \\
    \hline
    1000000 & 1 & 10 & False & 8.98 & 8.99 \\
    1000000 & 1 & 10 & True & 8.33 & 8.36 \\
    2000000 & 5 & 10 & False & 8.53 & 8.55 \\
    2000000 & 5 & 10 & True & 8.19 & 8.21 \\
    4000000 & 10 & 10 & False & 8.46 & 8.47 \\
    4000000 & 10 & 10 & True & 8.16 & 8.18 \\
    \hline
    \end{tabular}
    \caption{Prices for a call option on the maximum of 2 assets using a basis of $P\times P$ local functions  and parameters $S_0=(90, 90)$, $\sigma=(0.2,0.2)$, $\rho=0$, $\delta=(0.1,0.1)$, $T=3$, $r=0.05$, $K=100$, $N=9$. The Longstaff-Schwartz algorithm give a price of $8.1$.}\label{tab:max-call_loc}
\end{table}

\begin{table}[htb!]
  \centering\begin{tabular}{ccccc}
    \hline
    $Q$ & $\bar{N}$ & Vanilla\phantom{$\Big|$} & $U_0^Q$ & $\hat{U}_0^Q$ \\
    \hline
    1000000 & 1 & False & 9.07 & 9.07 \\
    1000000 & 1 & True & 8.32 & 8.33 \\
    1000000 & 5 & False & 8.68 & 8.69 \\
    1000000 & 5 & True & 8.16 & 8.17 \\
    2000000 & 10 & False & 8.61 & 8.62 \\
    2000000 & 10 & True & 8.14 & 8.15 \\
    \hline
    \end{tabular}
    \caption{Prices for a call option on the maximum of 2 assets using a polynomial approximation of degree~$5$ and parameters $S_0=(90, 90)$, $\sigma=(0.2,0.2)$, $\rho=0$, $\delta=(0.1,0.1)$, $T=3$, $r=0.05$, $K=100$, $N=9$. The Longstaff-Schwartz algorithm give a price of $8.1$.}\label{tab:max-call_pol}
\end{table}

We see on both tables the interest of using European options in the hedging portfolio and of using $\bar{N}>1$ rebalancing dates between two exercising dates. Between the local basis and the polynomial one, the latter one gives slightly better results: with $\bar{N}=10$, we get a price that is very close to the one given by Longstaff and Schwartz.  

We have indicated in Figure~\ref{fig:max-call_loc} the P\&L associated with the corresponding hedging strategies. Again, we note a clear interest to include European options in the hedging portfolio. The P\&L variance is reduced from 7.3 (left plot) to 1.5 (center plot). We note that the P\&L distribution of the stock only portfolio is slightly shifted to the right. This is expected since its mean value is exactly the difference between $\hat{U}_0^Q$ and the Longstaff Schwartz price, which is equal to $8.55-8.1=0.45$ on this example.  {We have also plotted, on the right of Figure~\ref{fig:max-call_loc}, the P\&L obtained with the Monte Carlo method proposed by~\cite{WC}, when the delta hedging strategy is calculated once and for all at the initial time. Even if the magnitudes of the hedging errors are of the same order, the left plot is more peaked than the right one, which indicates a better hedge provided by our algorithm. Besides this, we believe that one strength of our algorithm is its ability to add any available financial instruments in the hedging portfolio. Of course, Monte Carlo methods such as the one of~\cite{WC} give also approximations of the Gamma matrix, but it is not always obvious to translate this matrix in terms of financial instruments (European options on single assets or on many assets with different strikes). Our algorithm gets around this difficulty and directly computes the quantity of financial instruments for the hedging strategy.}

\begin{figure}[h!]
  \begin{center}
    \includegraphics[width=0.32\textwidth]{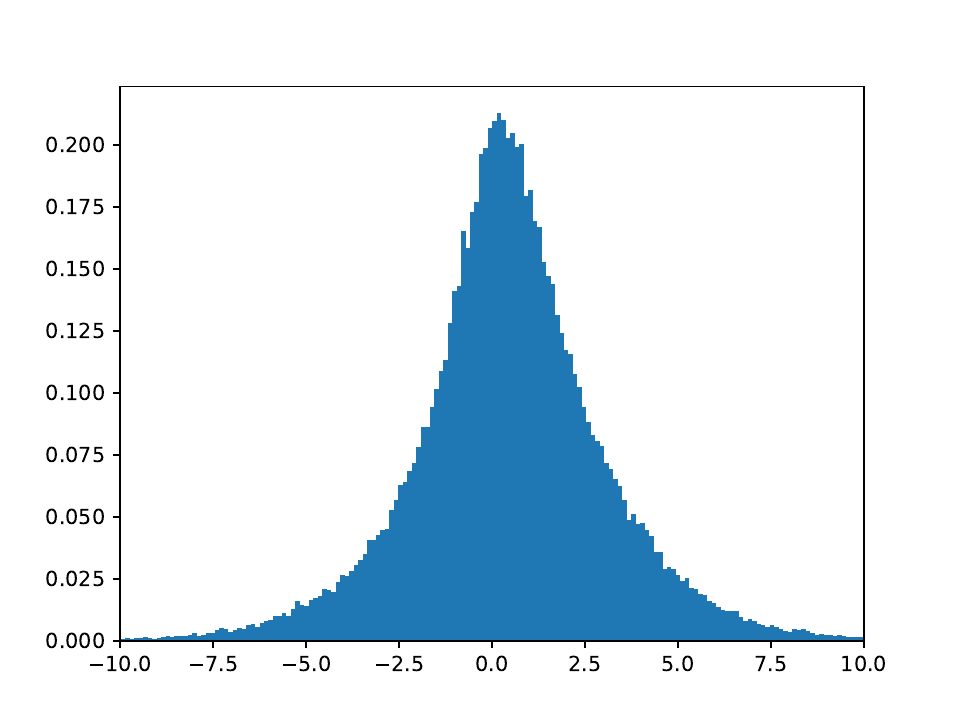}
    \includegraphics[width=0.32\textwidth]{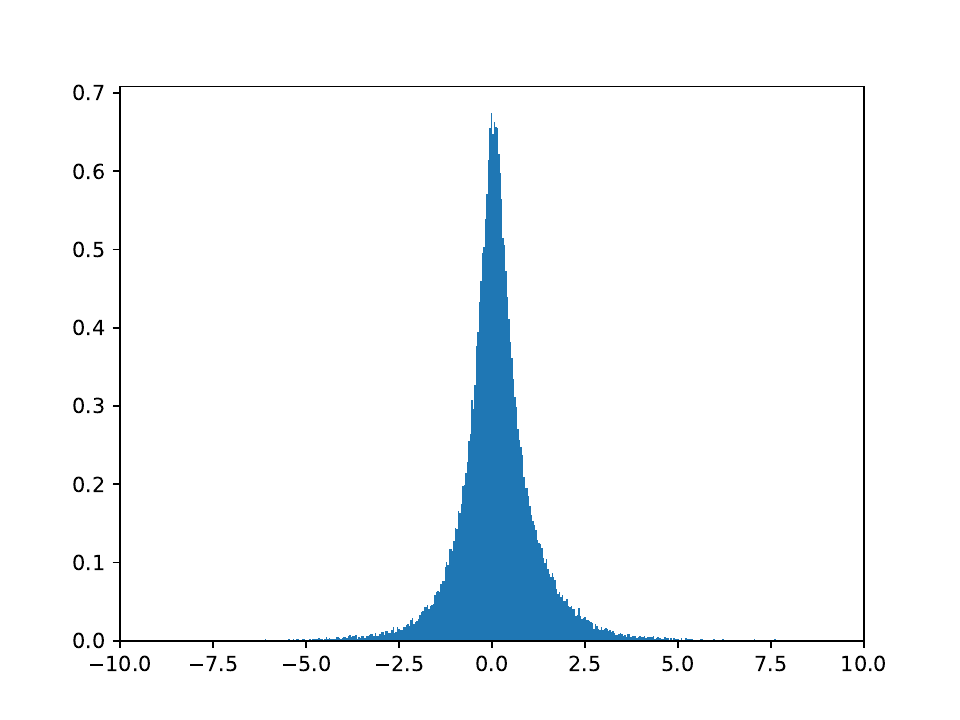} 
    \includegraphics[width=0.32\textwidth]{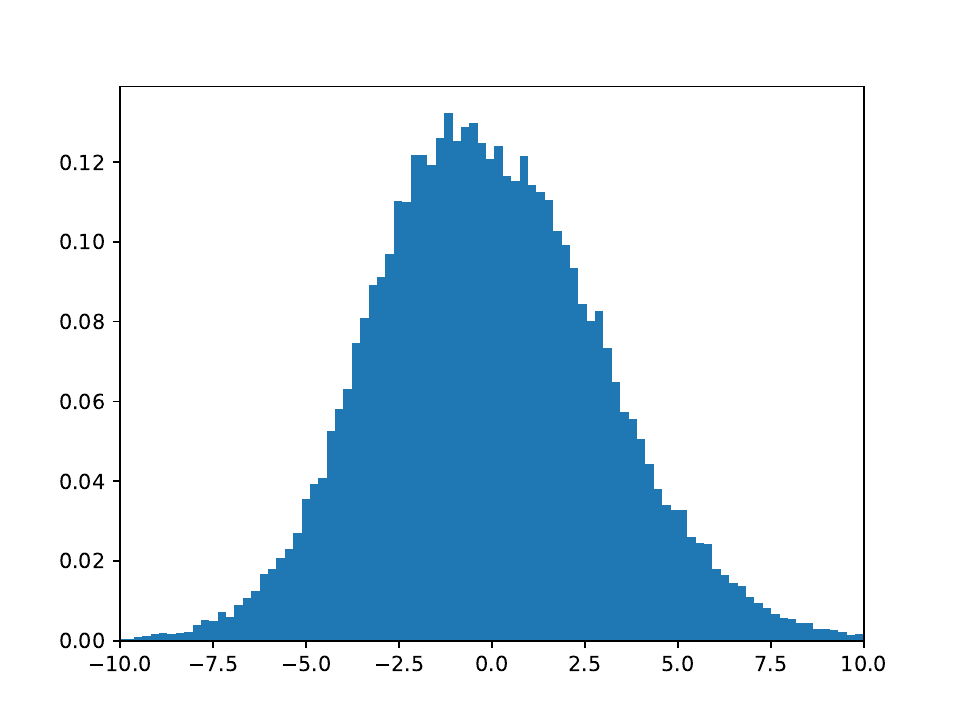}    
  \end{center}
  \caption{P\&L histograms of the hedging strategy for the Bermudan max-call option of Table~\ref{tab:max-call_loc} obtained with $\bar{N}=5$, $P=10$, $Q=2\times 10^6$ for the stock only strategy (left) and the strategy using extra European options (center). The right plot shows the result obtained with the delta hedging strategy calculated with~\cite{WC}.  }\label{fig:max-call_loc}
\end{figure}

\subsubsection{Min-put option}

We consider a put option on the minimum of two assets with payoff function $\Psi(S)=(K - \min(S^1_T, S^2_T))_+$. We take the same parameters as \cite[Table 4.5]{rogers-02}, but consider here a Bermudan option with 10 exercising dates instead of~50. Results are given in Table~\ref{tab:min-put_loc} for the local basis regression and in Table~\ref{tab:min-put_pol} for the polynomial basis. We obtain in Table~\ref{tab:min-put_loc} a price of 22.82, which is very close to 22.83 obtained in~\cite[Table 4.5]{rogers-02}.

Similarly to the max-call example, we observe that including the European option in the hedging portfolio leads to a significant improvement on its price.. Also, we observe that for the portfolio including European options, increasing the rebalancing frequency does not help much.
Contrary to our max-call example, the regression with the local basis gives here slightly better prices compared to the one obtained with polynomial functions.

\begin{table}[h!]
  \centering\begin{tabular}{cccccc}
    \hline
    $Q$ & $\bar{N}$ & $P$ & Vanilla\phantom{$\Big|$} & $U_0^Q$ & $\hat{U}_0^Q$ \\
    \hline
    1000000 & 1 & 10 & False & 23.53 & 23.54 \\
    1000000 & 1 & 10 & True & 22.83 & 22.86 \\
    2000000 & 5 & 10 & False & 23.06 & 23.10 \\
    2000000 & 5 & 10 & True & 22.77 & 22.87 \\
    4000000 & 10 & 10 & False & 22.98 & 23.03 \\
    4000000 & 10 & 10 & True & 22.75 & 22.82 \\
    \hline
  \end{tabular}
  \caption{Prices for a put on the minimum of two assets with a basis of $P^2$ local functions and parameters $S_0=(120,100)$, $\sigma=(0.4,0.8)$,  $\rho=0$, $\delta=0$, $T=0.5$, $r=0.06$, $K=100$, $N=10$. The Longstaff-Schwartz algorithm with a polynomial approximation of degree~$5$ yields a price of $22.6$.}\label{tab:min-put_loc}
\end{table}

\begin{table}[h!]
  \centering\begin{tabular}{cccccc}
    \hline
    $Q$ & $\bar{N}$  & Vanilla\phantom{$\Big|$} & $U_0^Q$ & $\hat{U}_0^Q$ \\
    \hline
    1000000 & 1 & False & 24.12 & 24.14 \\
    1000000 & 1 & True & 22.93 & 22.93 \\
    1000000 & 5 & False & 23.93 & 23.98 \\
    1000000 & 5 & True & 22.85 & 22.90 \\
    2000000 & 10 & False & 23.89 & 23.93 \\
    2000000 & 10 & True & 22.85 & 22.88 \\
    \hline
  \end{tabular}
  \caption{Prices for a put on the minimum of two assets with a basis of polynomial functions of degree $5$ and parameters $S_0=(120,100)$, $\sigma=(0.4,0.8)$,  $\rho=0$, $\delta=0$, $T=0.5$, $r=0.06$, $K=100$, $N=10$. The Longstaff-Schwartz algorithm with a polynomial approximation of degree $5$ yields a price of $22.6$.}\label{tab:min-put_pol}
\end{table}

We have drawn in Figure~\ref{fig:min-put_loc} the P\&L of the hedging strategy obtained with the local basis. We again observe that adding European options in the hedging strategy does not only improve the selling price but also reduces the variance of the P\&L. The variance is reduced from 36.6 to 4. between Figure~\ref{fig:min-put_loc_a} and~\ref{fig:min-put_loc_b}. In addition, we note that increasing the rebalancing frequency does not improve that much the P\&L distribution on the portfolio with European options: the variance is only reduced from 4 in Figure~\ref{fig:min-put_loc_b} to 2.96 in Figure~\ref{fig:min-put_loc_c} .

\begin{figure}[h!]
  \begin{subfigure}{0.33\textwidth}
    \includegraphics[width=\textwidth]{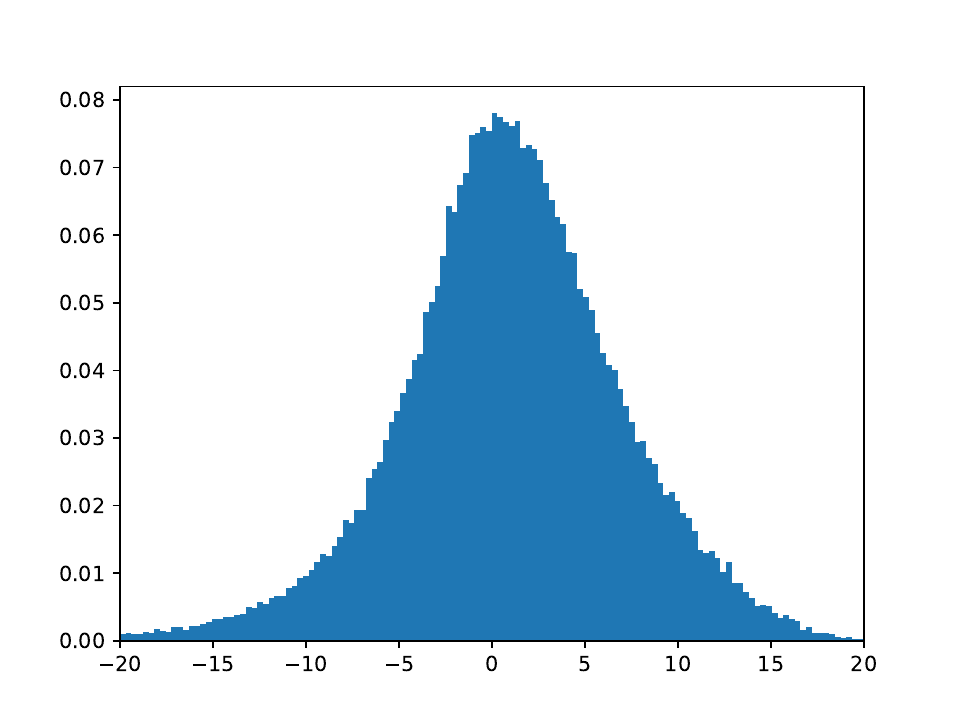}\caption{$\bar{N}=1$, $P=10$, $Q=10^6$.}\label{fig:min-put_loc_a}
  \end{subfigure}
  \begin{subfigure}{0.33\textwidth}
    \includegraphics[width=\textwidth]{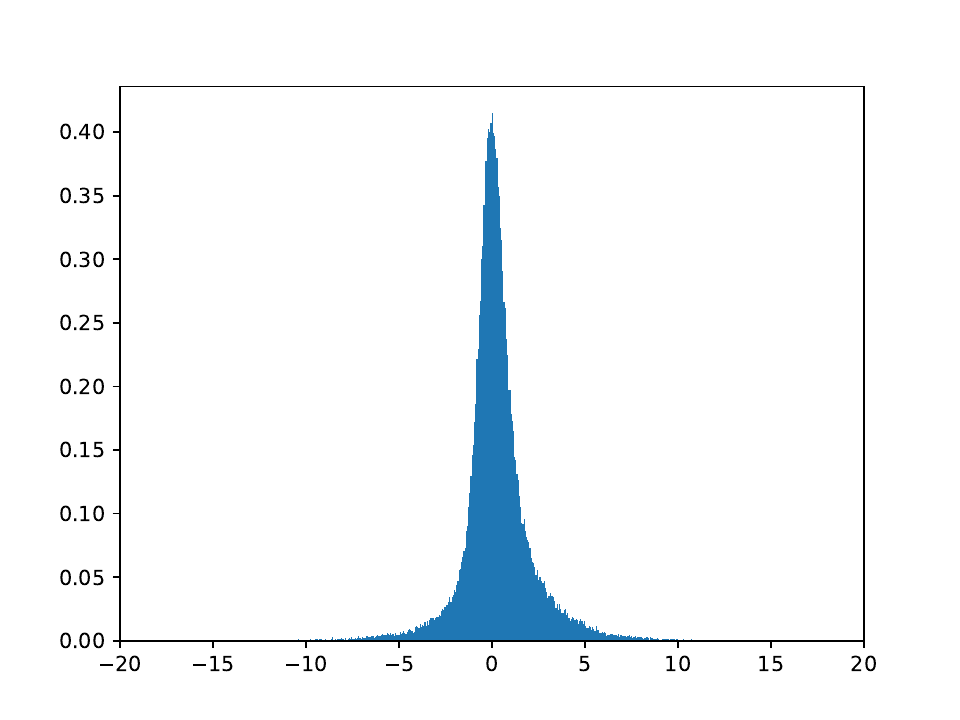}\caption{$\bar{N}=1$, $P=10$, $Q= 10^6$.}\label{fig:min-put_loc_b}
  \end{subfigure} 
  \begin{subfigure}{0.33\textwidth}
    \includegraphics[width=\textwidth]{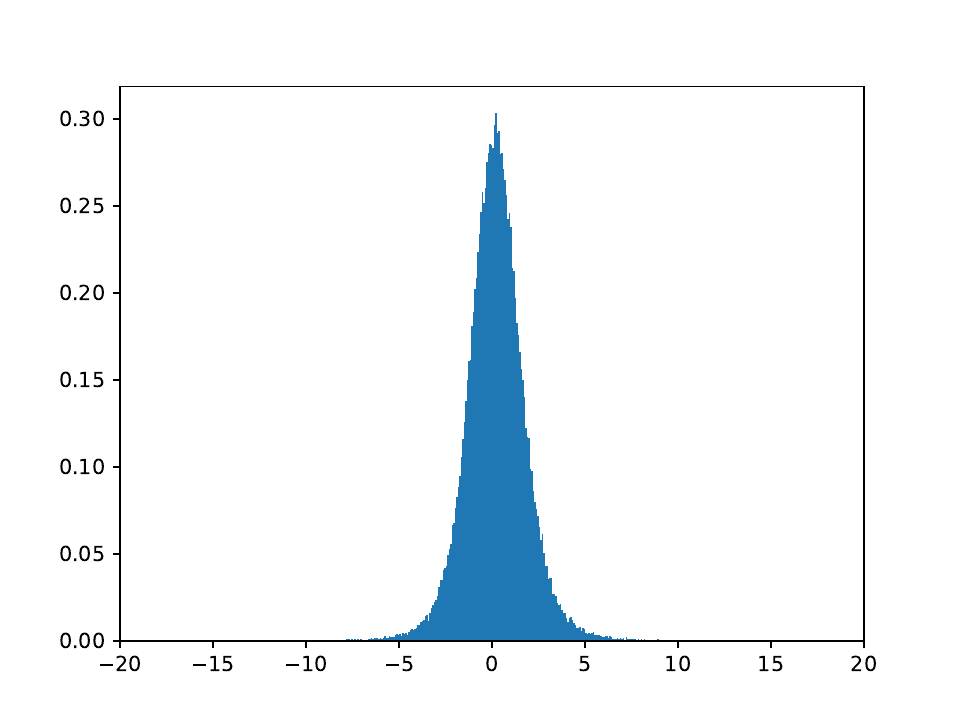}\caption{$\bar{N}=10$, $P=10$, $Q=4\times 10^6$.}\label{fig:min-put_loc_c}
\end{subfigure}
  \caption{P\&L histograms of the hedging strategy for the Bermudan min-put option of Table~\ref{tab:min-put_loc}  for the stock only strategy (Figure \ref{fig:min-put_loc_a}) and the strategy using extra European options (Figures \ref{fig:min-put_loc_b} and \ref{fig:min-put_loc_c}). }\label{fig:min-put_loc}
\end{figure}

\subsubsection{Basket option}

We consider a put option on a basket of assets with payoff function
\[ \Psi(S)=\left(K - \inv{d}\sum_{i = 1}^d S^i \right)_+, \  S \in \R^d. \]

We have indicated in Table~\ref{tab:basket_3d_locasset} the prices given by our pure dual algorithm when using the local basis. We get results that are fairly close to the Longstaff and Schwarz price with $\bar{N}=10$, $P=10$ and $Q=10^7$. Note that in dimension~3, the size of the local basis is $P^3$ and therefore $P$ cannot be taken too large in practice. The total number of parameters to estimate in the regression is $N\times \bar{N} \times P^3\times \bar{d}$, and we need to have $Q$ large enough which makes the computation intensive. We have got similar results with the polynomial regression that are not included here.

To get around this difficulty, we have considered a local basis of size~$P$ on the signed payoff  $K - \inv{d}\sum_{i = 1}^d S^i $, as explained in Appendix~\ref{App_basis}. The results are reported in Table~\ref{tab:basket_3d_locpayoff}.  We get better results with much less computational effort: with $\bar{N}=10$, $P=50$ and only $Q=10^6$, we get a price of $4.11$. Contrary to our previous examples, we note on this example that the use of European options in the hedging strategy does not really improve it. This is confirmed by the P\&L histograms in Figure~\ref{fig:basket_3d_locpayoff}, which shows that the P\&L variance is even slightly increased when including European options in the hedging portfolio.  
This may be due to our choice of European options in the portfolio, but we leave the investigation of other choices for further research. Note that this exhibits an other interest of our algorithm: it can be used as a way to select the assets to be included in the hedging portfolio.

\begin{table}[htb!]
  \centering\begin{tabular}{cccccc}
    \hline
    $Q$ & $\bar{N}$ & $P$ & Vanilla\phantom{$\Big|$} & $U_0^Q$ & $\hat{U}_0^Q$ \\
    \hline
    2000000 & 1 & 10 & False & 4.35 & 4.37 \\
    2000000 & 1 & 10 & True & 4.20 & 4.25 \\
    5000000 & 5 & 10 & False & 4.15 & 4.19 \\
    5000000 & 5 & 10 & True & 4.08 & 4.16 \\
    10000000 & 10 & 10 & False & 4.12 & 4.16 \\
    10000000 & 10 & 10 & True & 4.07 & 4.15 \\
    \hline
    \end{tabular}
  \caption{Prices for a basket put option in dimension $d=3$ using a basis of local functions with $K = S_0 = 100$, $T = 1$, $r=0.05$, $\sigma^i=0.2$, $\rho = 0.3$ and $10$ exercising dates. The Longstaff Schwartz algorithm with a polynomial approximation of order $3$ gives $4.03$.}\label{tab:basket_3d_locasset}
\end{table}

\begin{table}[htb!]
    \centering\begin{tabular}{cccccc}
    \hline
    $Q$ & $\bar{N}$ & $P$ & Vanilla\phantom{$\Big|$} & $U_0^Q$ & $\hat{U}_0^Q$ \\
    \hline
    100000 & 1 & 50 & False & 4.32 & 4.34 \\
    100000 & 1 & 50 & True & 4.29 & 4.32 \\
    250000 & 5 & 50 & False & 4.11 & 4.15 \\
    250000 & 5 & 50 & True & 4.09 & 4.16 \\
    500000 & 10 & 50 & False & 4.07 & 4.11 \\
    500000 & 10 & 50 & True & 4.05 & 4.13 \\
    1000000 & 10 & 50 & False & 4.08 & 4.11 \\
    1000000 & 10 & 50 & True & 4.07 & 4.12 \\
    \hline
  \end{tabular}
  \caption{Prices for a put option on $3$-dimensional basket using a basis of local functions of the signed payoff (see Equation~\eqref{def_upij_localpayoff})  with $K = S_0 = 100$, $T = 1$, $r=0.05$, $\sigma^i=0.2$, $\rho = 0.3$ and $10$ exercising dates. The Longstaff Schwartz algorithm with a polynomial approximation of order $3$ gives $4.03$.}\label{tab:basket_3d_locpayoff}
\end{table}

\iffalse
\begin{table}[htbp!]
  \centering\begin{tabular}{ccccc}
    \hline
    $Q$ & $\bar{N}$ & Vanilla \phantom{$\Big|$} & $U_0^Q$ & $\hat{U}_0^Q$  \\
    \hline
    100000 & 1 & False & 4.42 & 4.43 \\
    100000 & 1 & True & 4.25 & 4.28 \\
    100000 & 5 & False & 4.19 & 4.30 \\
    100000 & 5 & True & 4.02 & 4.23 \\
    500000 & 5 & False & 4.26 & 4.28 \\
    500000 & 5 & True & 4.15 & 4.19 \\
    1000000 & 10 & False & 4.24 & 4.26 \\
    1000000 & 10 & True & 4.14 & 4.18 \\
    \hline
  \end{tabular}
  \caption{Prices for a put option on $3$-dimensional basket using a polynomial approximation of order $5$ with $K = S_0 = 100$, $T = 1$, $r=0.05$, $\sigma^i=0.2$, $\rho = 0.3$ and $10$ exercising dates. The Longstaff Schwartz algorithm with a polynomial approximation of order $3$ gives $4.03$.}
\end{table}
\fi

\begin{figure}[h!]
  \begin{center}
  \includegraphics[width=0.4\textwidth]{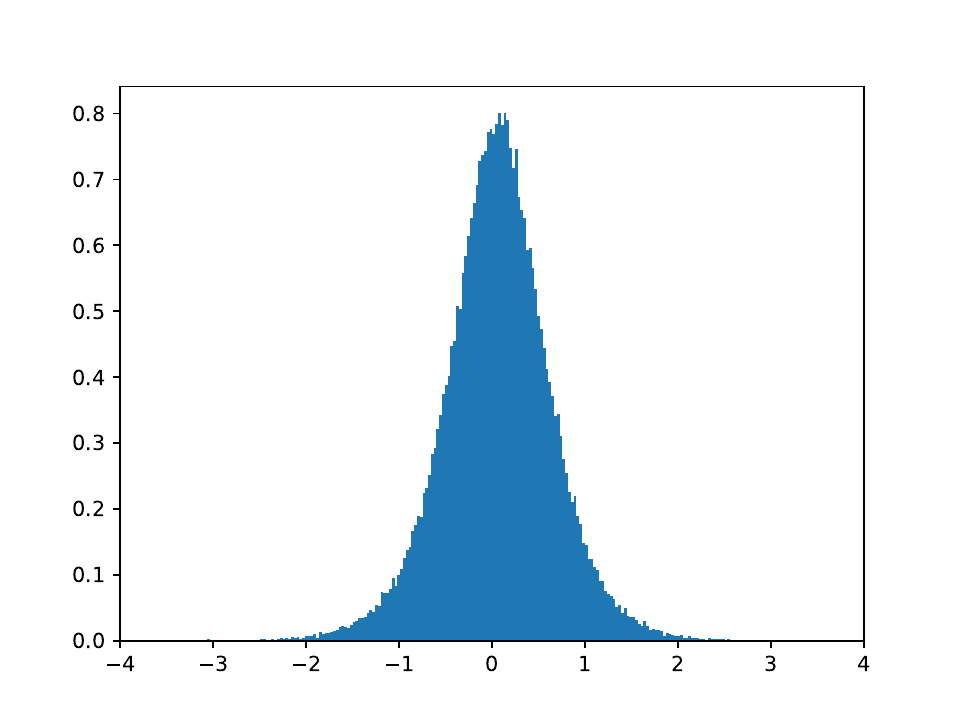}
  \includegraphics[width=0.4\textwidth]{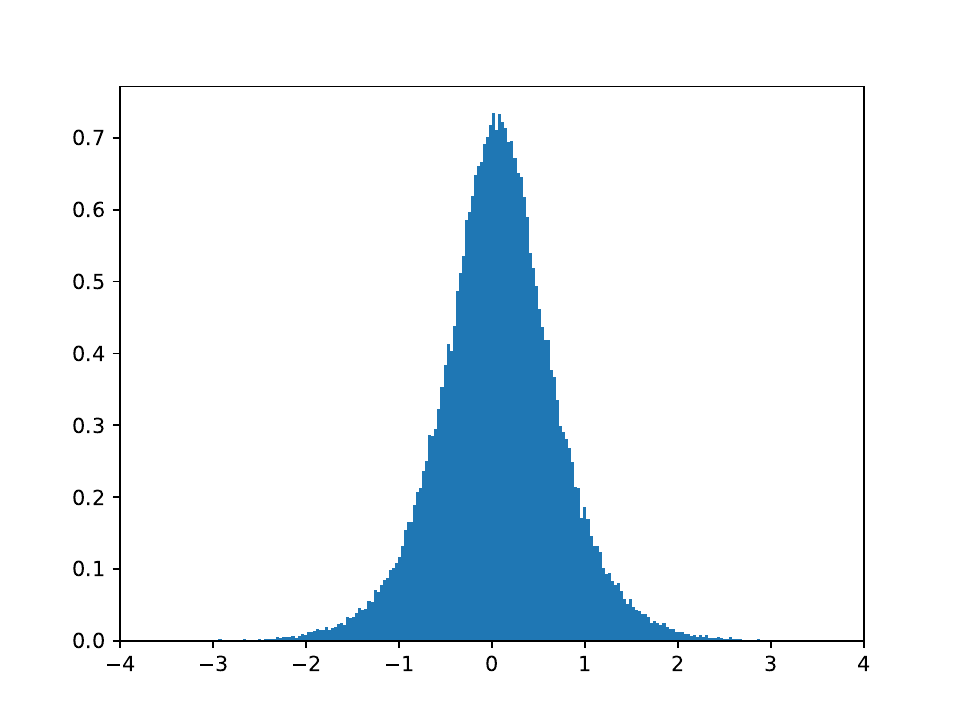}  
  \end{center}
  \caption{P\&L histograms of the hedging strategy for the Bermudan max-call option of Table~\ref{tab:basket_3d_locpayoff} obtained with $\bar{N}=10$, $P=50$, $Q=10^6$ for the stock only strategy (left) and the strategy using extra European options (right). }\label{fig:basket_3d_locpayoff}
\end{figure}

~\\
\noindent {\bf Data availability statement:} Data sharing not applicable to this article as no datasets were generated or analysed during the current study.\\
\noindent {\bf Conflict of interest:} The authors have no conflict of interest to declare.

%\clearpage 
\appendix
\section{Presentation of the basis functions}\label{App_basis} 

In this appendix, we present in more details the different choices that we have used for the functions $u^p_{i,j}$. These function are involved in Equation~\eqref{def_mg_incr} to define the subvector space of martingale increments.

For $d', P\in \N^*$ and ${\bf p}=(p_1,\dots,p_{d'})\in \{0,\dots, P-1\}^{d'}$, we define  the function 
\[
u^{d',\bf p}(x) = \prod_{i=1}^{d'} \ind{\left[\frac{p_i}{P}, \frac{p_i + 1}{P}\right)}(x_i), \, \text{ with } x=(x_1,\dots,x_{d'})\in \R^{d'},
\]
These functions are the indicator functions of the elementary hypercubes. 

\subsubsection*{Local regression basis}
For each asset $S^k$ and each time $t_{i,j}$ defined by~\eqref{def_subticks}, we consider a mapping  $\varphi^k_{i,j}:\R_+\to [0,1)$ so that $\varphi_{i,j}:\R_+^d\to [0,1)^d$. In our numerical experiments, we take for $\varphi^k_{i,j}$ the cumulative distribution function of a lognormal distribution with mean $\E(S^k_{t_{i,j}})$ and variance $\Var(S^k_{t_{i,j}})$. These quantities are estimated empirically on the samples at time $t_{i,j}$. Then, we define 
for $\Bar{P}=P^d$ and $1\le p\le \bar{P}$, 
\begin{equation}\label{def_upij_local} u^p_{i,j}=u^{d,{\bf p}} \circ \varphi_{i,j},\end{equation}
with $p-1=p_1+p_2\times P+\dots p_d \times P^{d-1}$ and ${\bf p}=(p_1,\dots,p_d)$ (here, the $d$-tuple ${\bf p}\in \{0,\dots,P-1 \}^{d}$ is represented by one natural number~$p\in \{1,\dots,\bar{P}\}$). 

Now, let us make precise how to solve the linear system~\eqref{linear_system_alpha} with this choice. We use the local property of the regression basis $u^{d,{\bf p}}(x)u^{d,{\bf p'}}(x)=0$ for ${\bf p}\not = {\bf p'}$. This gives for each $i$, $j$ and ${\bf p}$ the linear system
\begin{align*}\forall k \in \{1,\dots, \bar{d} \}, \ \sum_{k'=1}^{\bar{d}}  &\alpha_{i,j}^{{\bf p},k'} \E[(X^{{\bf p},k}_{t_{i,j}}-X^{{\bf p},k}_{t_{i,j-1}})(X^{{\bf p},k'}_{t_{i,j}}-X^{{\bf p},k'}_{t_{i,j-1}})]\\=&\E\left[(X^{{\bf p},k}_{t_{i,j}}-X^{{\bf p},k}_{t_{i,j-1}}) \left(\max_{i + 1\le m \le N} \left\{ Z_{m} - \sum_{\ell=i+2}^{m} \alpha_\ell \cdot \Delta X_\ell \right\}  \right)\right].
\end{align*}
Thus, for each $i$ and $j$, solving the linear system~\eqref{linear_system_alpha} of size $P^{d}\times \bar{d}$ determining $\alpha_{i,j}^{{\bf \cdot},\cdot}$ boils down to solving $P^{d}$ linear systems of size~$\bar{d}$, each of them determining $\alpha_{i,j}^{{\bf p},\cdot}$. Then, the Monte-Carlo estimator of $\alpha_{i,j}^{{\bf p},k'}$, denoted by $\alpha_{i,j}^{{\bf p},k',Q}$ is obtained by solving
\begin{align*}\sum_{k'=1}^{\bar{d}}  &\alpha_{i,j}^{{\bf p},k',Q}  \frac 1 Q  \sum_{q=1}^Q [(X^{{\bf p},k,q}_{t_{i,j}}-X^{{\bf p},k,q}_{t_{i,j-1}})(X^{{\bf p},k',q}_{t_{i,j}}-X^{{\bf p},k',q}_{t_{i,j-1}})]\\=&\frac 1 Q \sum_{q=1}^Q(X^{{\bf p},k,q}_{t_{i,j}}-X^{{\bf p},k,q}_{t_{i,j-1}}) \left(\max_{i + 1\le m \le N} \left\{ Z^q_{m} - \sum_{\ell=i+2}^{m} \alpha^Q_\ell \cdot \Delta X_\ell^q \right\}  \right).
\end{align*}

\subsubsection*{Local regression basis on the signed payoff}

The main advantage of the local regression basis is that it makes it simpler to solve the linear system, as we have explained above. Its main drawback is to heavily suffer from the curse of dimensionality. To get around this difficulty, a possibility is to look for alternative choices. For example, when considering a basket option with payoff $\left(K- \frac 1d \sum_{k=1}^d S^k_n \right)_+$,
it is rather natural to consider a local basis associated with the signed payoff function. Namely, we define $\psi_{i,j}(s)= \varphi_{i,j}\left( K- \frac 1d \sum_{k=1}^d s^k\right)$, where $\varphi_{i,j}: \R \to [0,1]$. In practice, we use for $\varphi_{i,j}$ the cumulative distribution function of a normal distribution with mean $\E[K- \frac 1d \sum_{k=1}^d S^k_{t_{i,j}}]$ and variance $\Var(K- \frac 1d \sum_{k=1}^d S^k_{t_{i,j}})$
These quantities are estimated empirically on the samples at time $t_{i,j-1}$. Then, we define
\begin{equation}\label{def_upij_localpayoff} u^p_{i,j}=u^{1,p} \circ \psi_{i,j}, \ 1 \le p \le P=\bar{P}.  
\end{equation}
This choice allows us to take advantage of the local structure while keeping only a moderate number of basis functions. Thus, the linear system~\eqref{linear_system_alpha} can be solved from $P$ linear systems of size $\bar{d}$ (to be compared with the $P^d$ linear systems with the previous local basis).

\subsubsection*{Polynomial functions}

This is a rather classical choice. For $\eta \in \N$, there are $\bar{P}=\binom{d+\eta}{\eta}$ monomials of degree lower or equal to~${\eta}$. If $(u^p)_{1\le p\le \bar{P}}$ denotes this monomial basis, we simply take 
\begin{equation}\label{def_upij_polynomial}
  u^p_{i,j}=u^p\circ \varphi_{i,j},
\end{equation} where $\varphi_{i,j}:(\R_+)^k \to \R^k$ is a linear map that essentially maps the assets values to $[-1,1]$ i.e. $\P(\varphi^k_{i,j}(S^k_{t_{i,j}}) \not \in [-1,1]) \ll 1$ for each $k\in \{1,\dots,d\}$. In the Black-Scholes model~\eqref{eq:BS_model}, we have taken $\varphi^k_{i,j}(x)=\frac{x-C_{i,j}^{k,-}}{C_{i,j}^{k,+}-C_{i,j}^{k,-}}$ with $C_{i,j}^{k,\pm}=S^k_0 \exp\left( \left(r-\delta^k- \frac 12 (\sigma^k)^2\right)t_{i,j} \pm 4 \sqrt{t_{i,j}}\sigma^k\right)$ since $\P(W_1 \not \in [-4,4])\ll 1$.

\bibliographystyle{abbrvnat}
\bibliography{biblio.bib}

\end{document}